\documentclass{article}

\usepackage{authblk}
\usepackage{amsmath}
\usepackage{amssymb}
\usepackage{amsthm}
\usepackage{graphicx}
\usepackage{caption}
\usepackage{subcaption}
\usepackage[all]{xy}

\makeatletter
\def\@fnsymbol#1{\ensuremath{\ifcase#1\or *\or \ddagger\or \dagger\or
   \mathsection\or \mathparagraph\or \|\or **\or \dagger\dagger
   \or \ddagger\ddagger \else\@ctrerr\fi}}\makeatother

\theoremstyle{definition}
\newtheorem{defn}{Definition}[section]

\theoremstyle{plain}
\newtheorem{lem}[defn]{Lemma}
\newtheorem{thm}[defn]{Theorem}
\newtheorem{prop}[defn]{Proposition}

\DeclareMathOperator{\boundeds}{\mathcal{B}}
\DeclareMathOperator{\Tr}{Tr}
\DeclareMathOperator{\Hom}{Hom}
\DeclareMathOperator{\id}{id}
\DeclareMathOperator{\im}{im}
\DeclareMathOperator{\supp}{supp}
\let\Re\undefined
\DeclareMathOperator{\Re}{Re}

\newcommand{\ket}[1]{\left|#1\right\rangle}
\newcommand{\braket}[2]{\left\langle#1\middle|#2\right\rangle}
\newcommand{\norm}[1]{\left\|#1\right\|}
\newcommand{\homatinfty}[1]{[#1]^\infty}
\newcommand{\cohomatinfty}[1]{[#1]_\infty}

\title{Homological codes and abelian anyons}

\author[1]{P\'eter Vrana\thanks{\texttt{vranap@math.bme.hu}}}
\author[2]{M\'at\'e Farkas\thanks{\texttt{mate.frks@gmail.com}}}
\affil[1]{Department of Geometry, Budapest University of Technology and Economics, Egry J\'ozsef u. 1., 1111 Budapest, Hungary}
\affil[2]{Department of Theoretical Physics, Budapest University of Technology and Economics, Budafoki \'ut 8., 1111 Budapest, Hungary}

\begin{document}
\maketitle

\begin{abstract}
We study a generalization of Kitaev's abelian toric code model defined on CW complexes. In this model qudits are attached to $n$ dimensional cells and the interaction is given by generalized star and plaquette operators. These are defined in terms of coboundary and boundary maps in the locally finite cellular cochain complex and the cellular chain complex. We find that the set of frustration free ground states and the types of charges carried by certain localized excitations depend only on the proper homotopy type of the CW complex. As an application we show that the homological product of a CSS code with the infinite toric code has excitations with abelian anyonic statistics.
\end{abstract}

\section{Introduction}\label{sec:intro}

Generalizations of Kitaev's toric code \cite{Kitaev} to spaces of higher dimension appear in many contexts. Ref. \cite{FreedmanMeyerLuo} makes use of discretizations of Riemannian manifolds in order to translate results in systolic geometry to a construction of quantum LDPC codes with distance growing faster than the square root of the block length. The authors of ref. \cite{FreedmanHastings} approach problems in quantum complexity theory using a variant of the model on certain cell complexes. They also reformulate hypergraph product codes as homological codes constructed from products of chain complexes, which are in turn used in ref. \cite{BravyiHastings} to prove the existence of good CSS codes having stabilizer weights bounded by the square root of the block length. In that paper it is pointed out that the homological product of the toric code with some fixed code might exhibit anyonic excitations. Other examples include the 4D toric code \cite{topmemory} and a 3D candidate for self-correcting quantum memory \cite{Brell}.

The essential features of all these models\footnote{With the exception of the ``single sector theory'' of ref. \cite{BravyiHastings}.} can be summarized as follows. First, one takes a CW complex and puts qudits on its $n$-cells for some fixed $n\in\mathbb{N}$. For each $n-1$-cell we form an $X$-type stabilizer from its coboundary and for each $n+1$-cell we form a $Z$-type stabilizer from its boundary. These stabilizers commute and generate the stabilizer subgroup of a CSS code. It is also possible to introduce a local Hamiltonian acting on the qudits such that its ground states are exactly the vectors in the code subspace. There is a freedom in doing this, in this paper we prefer to take the negative of the sum of projections onto the $1$ eigenspaces of the stabilizers.

An important property of Kitaev's toric code model is that it has excited states which can be interpreted as collections of anyonic quasiparticles\cite{Kitaev}. In order to study these in a mathematically rigorous way, one needs to formulate the problem in the framework of local quantum physics \cite{DHR1,DHR2,Halvorson}. Moreover, since different charges are then identified with inequivalent representations of the algebra of observables, it is necessary to consider infinite systems. This is done for the square lattice on the plane in refs. \cite{NaaijkensLocalized,NaaijkensDuality,NaaijkensIndex}, where localized excitations are classified and their braiding properties are derived from first principles. Under mild additional conditions on the CW complex, the corresponding more general model can also be discussed in this framework, and this is precisely the aim of the present paper.

It turns out that many results from the plane can be generalized to the present setting. There are some crucial differences, though. First of all, in the plane there is a unique translational invariant ground state, but there are infinitely many other ground states as well. Here translational invariance does not make sense, instead we concentrate on those ground states which minimize each interaction term to get a reasonably small set of states. Second, in the plane there is a natural selection criterion requiring that charges be transportable and localized in infinite cones. In general one cannot speak about cones, therefore we do not know what the ``right'' selection criterion should be. Finally, in the plane it is possible to define braiding in a canonical way, but the definition involves relative orientations of cones, hence it is also not applicable.

The above difficulties prevent us from performing a complete DHR-type analysis. Instead, we reformulate some known properties and constructions, invoking the language of algebraic topology, and generalize them to our setting. This leads to a distinguished class of ground states and low-energy excitations which depend only on the proper homotopy type of the CW complex.

The paper is organized as follows. In section \ref{sec:results} we briefly state the main results in an informal way. In section \ref{sec:examples} we illustrate the concepts and results on some special cases, namely, the ferromagnetic Ising model on certain infinite graphs, the Kitaev model on surfaces, and homological product codes. Section \ref{sec:homology} starts the formal treatment by giving a brief introduction to some tools in the algebraic topology of non-compact CW complexes. Besides ordinary cellular homology and cohomology, locally finite cellular (co-)homology and (co-)homology groups at infinity are covered here as a preparation for later sections. In section \ref{sec:model} we introduce our model in the C*-algebraic framework. In section \ref{sec:ground} we classify frustration free ground states. Section \ref{sec:endomorphisms} generalizes the string operators of Kitaev's toric code model to our setting and studies excited states associated with them. In section \ref{sec:GNSreps} we discuss properties of the GNS representations corresponding to these states. Here we also introduce an invariant which can tell apart some inequivalent representations. Finally, in section \ref{sec:braiding} we introduce and calculate a variant of the braiding operators, which can be given the usual interpretation when the space is plane-like.

\section{Results}\label{sec:results}

The model studied in this paper can be described informally as follows (see section \ref{sec:model} for a formal treatment). Given a CW complex $E$, a natural number $n$ and a finite group $G$, we put the Hilbert space $\ell^2(G)$ on each $n$-cell of $E$. We let $\mathcal{E}_i$ denote the set of $i$ dimensional cells. Cells of various dimensions will be labelled by the symbols $e_\alpha,e_\beta,e_\sigma,\ldots$. We will make use of the generalised Pauli operators acting on $\ell^2(G)$ as
\begin{equation}
X^g\ket{h}=\ket{g+h}\text{ and }Z^\chi\ket{h}=\chi(h)\ket{h},
\end{equation}
where $g\in G$ and $\chi:G\to\mathbb{C}$ is an irreducible character. The Hamiltonian defining the system is
\begin{equation}
H=-\sum_{\alpha\in\mathcal{E}_{n-1}}A_\alpha-\sum_{\beta\in\mathcal{E}_{n+1}}B_\beta,
\end{equation}
where $A_\alpha$ and $B_\beta$ are generalized star and plaquette operators centered at $\alpha$ and $\beta$, respectively. These are in turn given as
\begin{equation}
A_\alpha=\frac{1}{|G|}\sum_{g\in G}X^{\partial^T(g e_\alpha)}\text{ and }B_\beta=\frac{1}{|G|}\sum_{\chi\in \hat{G}}Z^{\partial(\chi e_\beta)},
\end{equation}
where the notation $X^a$ for an $n$-chain $a$ means the product of $X^g$ operators acting at the site $\sigma$ when there is a term $g e_\sigma$ in $a$, and similarly for $Z$-type operators. $\partial$ and $\partial^T$ are the boundary and coboundary operators, respectively.

This model admits frustration free ground states (when the system is finite, every ground state is frustration free). It turns out that the set of frustration free ground states is in bijection with the set of all states on a C*-algebra $\mathfrak{A}_{\mathrm{logical}}$ (see section \ref{sec:ground} for details). By analogy with the finite case (relevant in coding), we interpret these as logical states.

The structure of $\mathfrak{A}_{\mathrm{logical}}$ is determined by the $n$th homology ($H_n(E;\hat{G})$) and locally finite cohomology ($H_{\text{lf}}^n(E;G)$) groups and the canonical pairing between the two. It has a dense subalgebra with a basis consisting of unitaries $X^{[a]}Z^{[b]}$ where $[a]\in H_{\text{lf}}^n(E;G)$ and $[b]\in H_n(E;\hat{G})$, and they satisfy $(X^{[a]}Z^{[b]})(X^{[a']}Z^{[b']})=e^{2\pi i\langle b,a'\rangle}(X^{[a+a']}Z^{[b+b']})$. These properties determine $\mathfrak{A}_{\mathrm{logical}}$ essentially uniquely, see Proposition \ref{prop:logicalalgebra}.

When $H_n(E;\hat{G})$ and $H_{\text{lf}}^n(E;G)$ are finite groups, $\mathfrak{A}_{\mathrm{logical}}$ is finite dimensional, and the unitaries above form a basis. In this case $\mathfrak{A}_{\mathrm{logical}}\simeq\mathbb{C}^{2^c}\otimes\boundeds(\mathbb{C}^{2^q})$ where $c$ and $q$ are related to the degenerate and nondegenerate parts of the pairing and satisfy $2^{2q+c}=|H_n(E;\hat{G})\times H_{\text{lf}}^n(E;G)|$. For a precise statement, see Theorem \ref{thm:logicalstates}. This result means that the logical state can store $c$ classical bits and $q$ qubits of information.

The bijection between frustration free ground states and logical states respects the convex structure. In fact, a frustration free ground state is pure iff the corresponding logical state is pure (Proposition \ref{prop:purestrongground}). If we pass to the GNS representations, we find therefore that irreducible representations correspond to irreducible ones. Moreover, the GNS representation of a frustration free ground state is a factor iff the GNS representation of the corresponding logical state is a factor. This also implies that quasiequivalence of frustration free ground states is reflected in the quasiequivalence of logical states (Theorem \ref{thm:stronggroundGNS}).

Excited states can be created by composing a frustration free ground state with an endomorphism of the algebra of observables. Let $\gamma$ be a locally finite $n$-chain and $\delta$ an $n$-cochain. For any finite $K\subseteq\mathcal{E}_n$ we denote by $\gamma_K$ and $\delta_K$ their restrictions to $K$, i.e. the chain and locally finite cochain obtained by omitting any term supported outside $K$. Then an endomorphism can be formed as
\begin{equation}
\rho_{(\gamma,\delta)}:A\mapsto\lim_{K\to\mathcal{E}_n}Z^{\gamma_K}X^{\delta_K}AX^{-\delta_K}Z^{-\gamma_K}.
\end{equation}
If $\omega_0$ is a frustration free ground state then $\omega_{(\gamma,\delta)}:=\omega_0\circ\rho_{(\gamma,\delta)}$ minimizes those terms in the Hamiltonian which are not centered in the support of $\partial\gamma$ and $\partial^T\delta$. The increase in energy is precisely the sum of the sizes of these supports. This construction can also produce ground states that are not frustration free (Proposition \ref{prop:toground}).

In the planar Kitaev model anyonic excitations are described by GNS representations which are equivalent to the ground state when restricted to the complement of an infinite cone. For an arbitrary space we do not know what the ``right'' generalisation of this selection criterion is, but states of the form $\omega_{(\gamma,\delta)}$ with $|\supp\partial\gamma|<\infty$, $|\supp\partial^T\delta|<\infty$ seem to be the closest analogues of such states, therefore we will call them localized states. In this case $\gamma$ and $\delta$ can be thought of as representatives of homology and cohomology classes at infinity. Interestingly, the equivalence class of the GNS representation only depends on the classes $\homatinfty{\gamma}\in H^\infty_{n-1}(E;\hat{G})$ and $\cohomatinfty{\delta}\in H_\infty^n(E;G)$ (Proposition \ref{thm:equivalent}).

In the other direction, for these states it is possible to introduce an invariant in the form of a unitary representation of $H^\infty_n(E;\hat{G})\times H_\infty^{n-1}(E;G)$ in the center of the GNS representation. For $\cohomatinfty{c}\in H_\infty^{n-1}(E;G)$ and $\homatinfty{d}\in H^\infty_n(E;\hat{G})$ the representing operator is defined as
\begin{equation}
P_\omega(\homatinfty{d},\cohomatinfty{c}):=\lim_{K_\pm\to\mathcal{E}_{n\pm 1}}\pi_\omega\left(X^{\partial^T(c-c_{K_-})}Z^{\partial(d-d_{K_+})}\right),
\end{equation}
where the limit is understood in the strong operator topology (see Theorem \ref{thm:polarization}). We call $P_\omega$ the polarization of the state.

The definition of $P_\omega$ makes sense for states that are close to a frustration free ground state when restricted to a sufficiently distant compact contractible region (Definition \ref{def:astrong}). If $\homatinfty{\gamma}\in H^{\infty}_{n-1}(E;\hat{G})$ and $\cohomatinfty{\delta}\in H_\infty^n(E;G)$ then the polarization of $\omega'=\omega\circ\rho_{(\gamma,\delta)}$ is isomorphic to $P_\omega\otimes V_{(\gamma,\delta)}$ where $V_{(\gamma,\delta)}$ is one dimensional and $(\homatinfty{d},\cohomatinfty{c})$ acts as $e^{2\pi i(\langle\homatinfty{\gamma},\cohomatinfty{c}\rangle+\langle\homatinfty{d},\cohomatinfty{\delta}\rangle)}$ (Proposition \ref{prop:polarization-endomorphism}).

When $\homatinfty{\gamma}=\homatinfty{\gamma'}$ and $\cohomatinfty{\delta}=\cohomatinfty{\delta'}$ holds, one can find an $n$-chain $\hat{\gamma}$, a locally finite $n$-cochain $\hat{\delta}$, a locally finite $n+1$-chain $p$ and an $n-1$-cochain $q$ such that $\gamma'=\gamma-\hat{\gamma}+\partial p$ and $\delta'=\delta-\hat{\delta}+\partial^Tq$. Using these is it possible to construct a net in (the GNS representation of) the observable algebra converging to a unitary intertwiner (or charge transporter) between $\rho_{(\gamma,\delta)}$ and $\rho_{(\gamma',\delta')}$ in the strong operator topology (Proposition \ref{prop:transporter}).

Crucially, if we apply an endomorphism corresponding to a localized excitation to every element in this net, the result will still converge. This makes it possible to define braiding operations using such nets. However, the braiding operator will in general depend on the choice of $\hat{\gamma},\hat{\delta},p,q$ (see Propositions \ref{prop:braiding} and \ref{prop:freedominbraiding}).

\section{Examples}\label{sec:examples}

\subsection{Ferromagnetic Ising model}\label{sec:ising}
As the first example we consider the ferromagnetic Ising model with nearest neighbour interactions on a locally finite graph (i.e. a 1-cell joining each pair of neighbouring $0$-cells) with $c<\infty$ components and finitely many ends (see Figure \ref{fig:Isingstar}). This corresponds to taking $n=0$ and $G=\mathbb{Z}_2$.\footnote{We identify $\mathbb{Z}_d$ with its dual using the ring multiplication and the map $\mathbb{Z}_d\hookrightarrow\mathbb{Q}/\mathbb{Z}$ sending $1$ to (the equivalence class of) $1/d$.} Then we find that $H_0(E;\mathbb{Z}_2)$ is a vector space over $\mathbb{Z}_2$ with a basis corresponding to the set of connected compontents. On the other hand, $\dim_{\mathbb{Z}_2}H_{\text{lf}}^0(E;\mathbb{Z}_2)=c_0$ is the number of finite components. The pairing is nondegenerate on the finite components, therefore $\mathfrak{A}_{\mathrm{logical}}\simeq\mathbb{C}^{2^{c-c_0}}\otimes\boundeds(\mathbb{C}^{2^{c_0}})$ . In other words, the ground states can store one qubit for each finite component and one classical bit for each infinite component. It follows that quasiequivalence-classes of frustration free factor ground states are classified by finite bit strings of length $c-c_0$.

\begin{figure}
\centering
\includegraphics{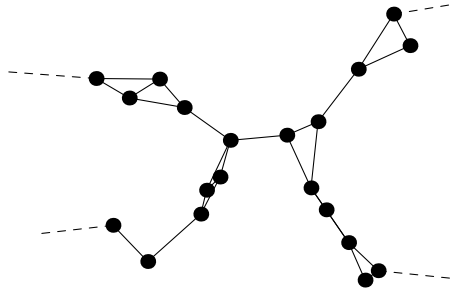}
\caption{The ferromagnetic Ising model can be defined on any locally finite graph. The dashed lines indicate that in this example the graph extends to infinity in four directions or, more formally, it has four ends.\label{fig:Isingstar}}
\end{figure}

Since $H^{\infty}_{-1}(E;\mathbb{Z}_2)=0$, there are no $Z$-type excitations. On the other hand, $H_\infty^0(E;\mathbb{Z}_2)$ is nontrivial when there are infinite components. In this case $\dim_{\mathbb{Z}_2}(E;\mathbb{Z}_2)$ is the number of ends of the graph. A charged sector may be constructed as follows. Take a finite part $K$ of $E$ such that its complement has $\dim_{\mathbb{Z}_2}(E;\mathbb{Z}_2)$ connected components. For each such component there is an endomorphism of type $X$ flipping all the spins in that component. This corresponds to the formal sum $\delta$ of all the vertices in the component. The coboundary $\partial^T\delta$ can only contain edges joining a vertex inside $K$ with one outside $K$, and the number of such edges is finite.

$H^\infty_0(E;\mathbb{Z}_2)$ is isomorphic to $H_\infty^0(E;\mathbb{Z}_2)$. With $K$ as above, in each component of its complement we can pick a semi-infinite path. The formal sum of its edges represents a $0$-homology class at infinity. For this representative the boundary consists of a single point in the component in question, therefore the pairing with the cohomology class at infinity corresponding to that component gives $1/2$, while with other components we get $0$.

Suppose that we start from the ground state with all spins pointing in the ``up'' direction and then flip all the spins in some component of the complement of $K$. Take a semi-infinite path representing a generator of $H^\infty_0(E;\mathbb{Z}_2)$ as above. The value of the corresponding polarization is therefore $e^{2\pi i\cdot\frac{1}{2}}=-1$ or $e^{2\pi i\cdot 0}=1$ depending on whether the spins are flipped from some point along the path or not.

The map $H^0(E;\mathbb{Z}_2)\to H_\infty^0(E;\mathbb{Z}_2)$ is described as follows. $H^0(E;\mathbb{Z}_2)$ has a basis labelled by the components of the graph, the induced map takes finite components to $0$ and each infinite component to the sum of its ends.

\subsection{Kitaev model}\label{sec:kitaev}
The second example is Kitaev's model, where the space is a surface and $n=1$. For simplicity, we take $G=\mathbb{Z}_2$, but similar results hold for any finite abelian group. In the coding theory context one takes a compact orientible surface of genus $g$ together with some cellular decomposition. In this case the relevant groups are 
$H_1(E;\mathbb{Z}_2)\simeq\mathbb{Z}_2^{2g}\simeq H_{\text{lf}}^1(E;\mathbb{Z}_2)$ with a nondegenerate pairing, and one gets $\mathfrak{A}_{\mathrm{logical}}\simeq\boundeds(\mathbb{C}^2)^{\otimes 2g}$. Since a compact CW complex has finitely many cells, the inclusion maps induce isomorphisms between the ordinary and locally finite (co-)homology groups. As a consequence, (co-)homology groups at infinity vanish and we get one GNS representation up to quasiequivalence. This is also clear from the C*-algebraic viewpoint, because in this case the algebra of observables is a full matrix algebra. From the point of view of physics, the reason is that charged excitations are always created in opposite pairs, hence the total charge is always trivial.

The situation is different when the surface is noncompact. The simplest example is $\mathbb{R}^2$ with the cell structure given by a square lattice (Figure \ref{fig:kitaevplane}), which has been investigated in \cite{NaaijkensLocalized,NaaijkensDuality,NaaijkensIndex}. In fact, many of our results are inspired by these papers. First note that $H_1(\mathbb{R}^2;\mathbb{Z}_2)\simeq H_{\text{lf}}^1(\mathbb{R}^2;\mathbb{Z}_2)\simeq 0$ implies that $\mathfrak{A}_{\mathrm{logical}}\simeq\mathbb{C}$ is trivial, therefore there is a unique frustration free ground state. This state was computed in ref. \cite{AFH}.
\begin{figure}
\centering
\includegraphics{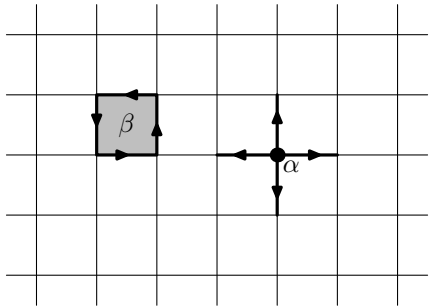}
\caption{Kitaev's toric code model on the plane. The CW structure is given by a rectangular lattice. Examples of star and plaquette operators indicated with arrows. ($n=1$)\label{fig:kitaevplane}}
\end{figure}

To understand charged states we need to compute the groups $H^\infty_0(\mathbb{R}^2,\mathbb{Z}_2)$ and $H_\infty^1(E;\mathbb{Z}_2)$. It turns out that both are isomorphic to $\mathbb{Z}_2$. The nonzero element of $H^\infty_0(\mathbb{R}^2,\mathbb{Z}_2)$ is represented by the formal sum of edges along any semi-infinite path. The boundary of such a representative consists of a single vertex, which can be interpreted as the location of a charged quasiparticle. Similarly, the nontrivial element of $H_\infty^1(\mathbb{R}^2;\mathbb{Z}_2)$ can be found using a semi-infinte dual path. A dual path is a sequence of plaquettes (2-cells) such that consecutive ones share an edge. The formal sum of the shared edges represents the generator, its coboundary is supported on the first plaquette.

To compute polarization operators in these excited states we need to find the nontrivial elements of $H_\infty^0(\mathbb{R}^2,\mathbb{Z}_2)$ and $H^\infty_1(\mathbb{R}^2;\mathbb{Z}_2)$. In the former case it is given by the sum of all vertices while in the latter case by the sum of all plaquettes. Given a finite subset $K_-$ of vertices, the sum of all vertices outside $K_-$ has a finite coboundary, which can be visualized as the set of edges joining $K_-$ to its complement. It is easier to imagine the situation when $K_-$ has a simple shape, e.g. the set of vertices in a large circle. In this case the coboundary is a dual path encircling the points inside, therefore it detects a semi-infinite path (i.e. picks up a factor of $-1$) iff its endpoint lies within the circle. As $K_-$ grows, eventually any point will be inside, therefore the corresponding operator will count the total charge carried by $Z$-type excitations.

Similarly, $H^\infty_1(\mathbb{R}^2;\mathbb{Z}_2)$ is generated by the class of the sum of all the plaquettes. The sum of plaquettes outside some finite subset $K_+$ has a finite boundary. Again, if $K_+$ is the set of plaquettes in a circle then this boundary is a closed path and the operator counts the number of $X$-type excitations inside (mod 2). The total charge is recovered as $K_+$ grows and eventually contains any plaquette.

It is possible to introduce a braiding on the charged sectors, and it turns out that the quasiparticles are in fact anyons. In our treatment this will follow from the more general setting discussed in sec. \ref{sec:homprod}

Slightly more generally, we can consider Kitaev's model on a noncompact surface with genus $0$ and finitely many ends. To obtain such a surface properly embedded in $\mathbb{R}^3$, we cut $k$ circular holes in the unit sphere and attach copies of the cylinder $S^1\times[0,\infty)$ along the boundaries of the holes. We denote this surface by $\Sigma_k$ (Figure \ref{fig:starfish}). A CW structure can be given using those on the sphere and on the cylinder. We remark that $\Sigma_1\simeq\mathbb{R}^2$, therefore we can get back the results for the plane as a special case.
\begin{figure}
\centering
\includegraphics{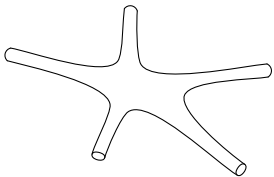}
\caption{Proper embedding of a non-compact surface obtained by cutting $k$ holes on a $2$-sphere and attaching semi-infinite cylinders. ($k=5$)\label{fig:starfish}}
\end{figure}

One can compute that $H_1(\Sigma_k;\mathbb{Z}_2)\simeq H_{\text{lf}}^1(\Sigma_k;\mathbb{Z}_2)\simeq \mathbb{Z}_2^{k-1}$ and the pairing between the two is trivial. Indeed, for any cylinder one can take a closed path (dual path) going around the $S^1$ factor. These are nontrivial, since none of them is the boundary (coboundary) of a finite sum of plaquettes (vertices). They generate the two groups, but in both cases there is one relation. The sum of all $k$ generators can be realized as the boundary of the sum of plaquettes (coboundary of the sum of vertices) on the sphere on which we have cut the holes. Therefore the ground states can store $2k-2$ classical bits.

Charged sectors can again be constructed using semi-infinite paths and dual paths. This time not all of these are equivalent, since any such path will eventually leave every cyilinder but one.\footnote{Note that we consider paths in the graph theoretic sense and therefore can only cross the sphere a finite number of times. The situation is different for injective continuous maps $[0,\infty)\to\Sigma_k$ (but not for proper ones).} Taking $k$ paths $\gamma_1,\ldots,\gamma_k$ and dual paths $\delta_1,\ldots,\delta_k$ going to infinity in the corresponding cylinder gives representatives of independent generators of $H^\infty_0(\Sigma_k;\mathbb{Z}_2)\simeq\mathbb{Z}_2^k$ and $H_\infty^1(\Sigma_k;\mathbb{Z}_2)\simeq\mathbb{Z}_2^k$, respectively.

The existence of inequivalent paths gives rise to an interesting phenomenon. We may assume that $\gamma_1,\ldots,\gamma_k$ start at the same point. Choose $i\neq j$ from $1,\ldots,k$ and consider the sum $\gamma_i+\gamma_j$. This path is infinite in both directions, therefore it has empty boundary, but it is not a boundary of any locally finite $2$-chain, i.e. it represents a nontrivial element in $H^{\text{lf}}_1(\Sigma_k;\mathbb{Z}_2)$. Therefore if we create both excitations, then it is not possible to tell the difference by measurements in compact contractible regions. Still, the state changes during this process, as can be seen using the polarization operators. There is also a phyisical picture behind this. Namely, it is possible to apply logical operations by creating a particle at infinity in one of the cylinders, moving towards the sphere and then again towards infinity, but along some other cylinder. The same can be done using dual paths.

Mathematically, this can be understood using the maps $H^{\text{lf}}_1(\Sigma_k;\mathbb{Z}_2)\to H^\infty_0(\Sigma_k;\mathbb{Z}_2)$ and $H^1(\Sigma_k;\mathbb{Z}_2)\to H_\infty^1(\Sigma_k;\mathbb{Z}_2)$ from the long exact sequences in eqs. \eqref{eq:longhomology} and \eqref{eq:longcohomology}. The groups $H^{\text{lf}}_1(\Sigma_k;\mathbb{Z}_2)$ and $H^1(\Sigma_k;\mathbb{Z}_2)$ are isomorphic to $\mathbb{Z}_2^{k-1}$ and are generated by the sums $\gamma_i+\gamma_j$ and $\delta_i+\delta_j$, respectively.

Finally, let us see how polarization operators look like. To each cylinder one can associate the element of $H_\infty^0(\Sigma_k;\mathbb{Z}_2)$ represented by the sum of its vertices and the element of $H^\infty_1(\Sigma_k;\mathbb{Z}_2)$ represented by the sum of its plaquettes. Truncating these to the complement of a large finite set of vertices (plaquettes) and taking the coboundary (boundary) results in a dual path (path) around the cylinder, and its expected value counts the number of paths (dual paths) in that cylinder attached to localized excitations (mod 2).

In the examples above, for $k>1$ there are frustration free ground states which give rise to GNS representations that are not quasiequivalent to each other. This is related to the classical information stored in the states. Inequivalent representations can arise in a different way for surfaces with infinite genus. As a simple example consider the surface $\Sigma$ of a thickened semi-infinite ladder embedded in $\mathbb{R}^3$ (Figure \ref{fig:infinitegenus}). This surface has one end, it is orientable and has infinite genus. In this case $H_1(\Sigma;\mathbb{Z}_2)$ and $H_{\text{lf}}^1(\Sigma;\mathbb{Z}_2)$ are isomorphic to the direct sum of countably many copies of $\mathbb{Z}_2$, and the pairing between the two is nondegenerate. Therefore $\mathfrak{A}_{\mathrm{logical}}$ is the tensor product of countably many $2\times 2$ matrix algebras\footnote{This can be given a precise meaning as the direct limit of tensor products of an increasing number of copies.}, which has uncountably many pairwise inequivalent irreducible representations.
\begin{figure}
\centering
\includegraphics{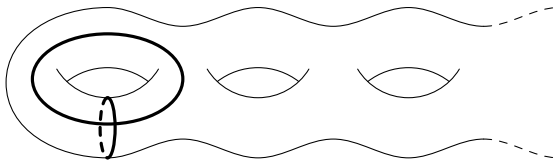}
\caption{An orientable surface with one end and infinite genus. Each hole gives rise to two of generators of $H_1(\Sigma;\mathbb{Z}_2)$ and of $H_\text{lf}^1(\Sigma;\mathbb{Z}_2)$. \label{fig:infinitegenus}}
\end{figure}

\subsection{Homological product codes}\label{sec:homprod}

The main motivation behind our studies comes from the homological product code construction. In ref. \cite{BravyiHastings} it was suggested that the homological product of the toric code with some fixed CSS code might support anyonic quasiparticles like the toric code itself. Given that the homological product corresponds to the Cartesian product of the underlying topological spaces, and the product of a plane with some other (nontrivial) space is no longer a surface, it is natural to consider the problem in our framework, possibly with $n>1$. Note that a CSS code does not determine a CW complex uniquely, therefore our starting point will be a finite (equivalently: compact) CW complex.

For this construction we need the following ingredients: a compact CW complex $F$, a natural number $n$ and a finite abelian group $G$. The space on which our system lives will be $E=F\times\mathbb{R}^2$ (where $\mathbb{R}^2$ is equipped with an arbitrary locally finite CW structure, e.g. the one given by the square lattice), and we will put a copy of $\ell^2(G)$ on each $n$-cell. The planar (abelian) Kitaev model is the special case when $F$ is a point and $n=1$. The problem of finding frustration free ground states and localized excitations is now reduced to a problem in algebraic topology: we need to understand various homology and cohomology groups of $E$. This can be done using the long exact sequences (eqs. \eqref{eq:longhomology} and \eqref{eq:longcohomology}). The result is stated in \cite[Example 3.10]{ends}, but for the calculations we need an explicit description of the group elements, which we give here.

Since $\mathbb{R}^2$ is contractible, $H_i(E;\hat{G})\simeq H_i(F;\hat{G})$ and $H^i(E;G)\simeq H^i(F;G)$ for any $i$. There is a canonical isomorphism induced by the homotopy equivalence $E=F\times\mathbb{R}^2\to F$ which is the projection onto the first factor. Its homotopy inverse is the inclusion $F\to F\times\{(0,0)\}\subseteq E$.\footnote{Here we assume without loss of generality that the origin is a $0$-cell.}

The inclusion of the cellular chain complex into the locally finite cellular chain complex induces trivial maps on homology. To see this, let $f$ be an $i$-cycle on $F$ with $\hat{G}$-coefficients and $\Gamma$ a semi-infinite path starting at the origin, thought of as an element in the locally finite chain complex of $\mathbb{R}^2$ with coefficients in $\mathbb{Z}$. Then $(-1)^{i}f\otimes\Gamma$ is a locally finite $i+1$-chain on $E$ with $\hat{G}$-coefficients and $\partial((-1)^{i}f\otimes\Gamma)=(-1)^{i}(\partial f)\otimes\Gamma+f\otimes\partial\Gamma=f\otimes\partial\Gamma$. Therefore the image of $f$ is a boundary. The map induced on cohomology groups is also trivial and the proof is similar.

$H^{\text{lf}}_i(E;\hat{G})$ is isomorphic to $H_{i-2}(F;\hat{G})$. If $p$ is the formal sum of all $2$-cells (with the appropriate signs $\pm1\in\mathbb{Z}$ so that $\partial p=0$) and $f$ is an $i-2$-cycle in $F$ with $\hat{G}$-coefficients then $\partial(f\otimes p)=(\partial f)\otimes p+(-1)^{i-2}f\otimes\partial p=0$, i.e. $f\otimes p$ is a locally finite cycle. When $f=\partial g$, $\partial(g\otimes p)=f\otimes p$ is a locally finite boundary, therefore we get a well-defned map $H_{i-2}(F;\hat{G})\to H^{\text{lf}}_i(E;\hat{G})$. It can be proved that this is an isomorphism. There is a similar isomorphism $H_{\text{lf}}^i(E;\hat{G})\simeq H^{i-2}(F;\hat{G})$, under which an $i-2$-cocycle $f$ of $F$ corresponds to $f\otimes p$ with $p$ a single $2$-cell of $\mathbb{R}^2$.

Now we can describe the logical algebra. For this we need $H_n(E;\hat{G})\simeq H_n(F;\hat{G})$ and $H_{\text{lf}}^n(E;G)\simeq H^{n-2}(F;\hat{G})$. The pairing between the two is identically $0$, therefore $\mathfrak{A}_{\mathrm{logical}}$ is commutative. Its dimension is equal to $|H_n(F;\hat{G})|\cdot|H^{n-2}(F;\hat{G})|$, therefore the ground states can store
\begin{equation}
\log_2|H_n(F;\hat{G})|+\log_2|H^{n-2}(F;G)|
\end{equation}
classical bits of information.

Next we need the groups $H^\infty_{i-1}(E;\hat{G})$ and $H_\infty^i(E;G)$. By the long exact sequence for homology, the induced map $H^{\text{lf}}_i(E;\hat{G})\to H^\infty_{i-1}(E;\hat{G})$ is injective and the map $H^\infty_{i-1}(E;\hat{G})\to H_{i-1}(E;\hat{G})$ is surjective, while their composition is $0$. Actually, $H^\infty_{i-1}(E;\hat{G})\simeq H^{\text{lf}}_i(E;\hat{G})\oplus H_{i-1}(E;\hat{G})$. The $H^{\text{lf}}_i(E;\hat{G})$ part of $H^\infty_{i-1}(E;\hat{G})$ is related to the different logical states (when $i=n$), therefore we concentrate on the other summand, which is isomorphic to $H_{i-1}(E;\hat{G})\simeq H_{i-1}(F;\hat{G})$.

Let $f$ be an $i-1$-cycle on $F$ and $\Gamma$ a semi-infinite path on $\mathbb{R}^2$ as before. Then $\partial(f\otimes\Gamma)=(-1)^{i-1}f\otimes\partial\Gamma$ has finite support, therefore $\gamma=f\otimes\Gamma$ represents an element of $H^\infty_{i-1}(E;\hat{G})$. If $f=\partial g$ then $f\otimes\Gamma=\partial(g\otimes\Gamma)-(-1)^i g\otimes\partial\Gamma$ is the sum of a locally finite boundary and a cycle with finite support, therefore it is trivial in $H^\infty_{i-1}(E;\hat{G})$. This gives a well-defined map $H_{i-1}(F;\hat{G})\to H^\infty_{i-1}(E;\hat{G})$. Similarly, if $\Delta$ is a semi-infinite dual path in $\mathbb{R}^2$ (more precisely, the corresponding cochain with coefficients $\pm1\in\mathbb{Z}$), then the map $g\mapsto g\otimes\Delta$ induces an injective map $H^{i-1}(F;\hat{G})\to H^\infty_i(E;\hat{G})$.

To find the possible charges, we need $H^\infty_{n-1}(F\times\mathbb{R}^2;\hat{G})\simeq H_{n-1}(F;\hat{G})\oplus H_{n-2}(F;\hat{G})$ and $H_\infty^n(F\times\mathbb{R}^2;G)\simeq H^{n-1}(F;G)\oplus H^n(F;G)$. The interesting parts are the summands $H_{n-1}(F;\hat{G})$ and $H^{n-1}(F;G)$, which give rise to excited states localized in cone-like regions. Any element from these subgroups is represented by a locally finite chain (cochain) supported inside any specified infinite cone and therefore these are localized and transportable. As usual, the total charge can be detected with polarization operators $P(\homatinfty{c},\cohomatinfty{d})$ with $\homatinfty{c}\in H^\infty_n(F\times\mathbb{R}^2;\hat{G})\simeq H_{n-1}(F;\hat{G})\oplus H_n(F;\hat{G})$ and $\cohomatinfty{d}\in H_\infty^{n-1}(F\times\mathbb{R}^2;G)\simeq H^{n-1}(F;G)\oplus H^{n-2}(F;G)$. Here the first summands are related to the charges considered above, while frustration free ground states can be distinguished using the second summands.

Since $E=F\times\mathbb{R}^2$ is essentially planar, it is also possible to introduce braiding operators in a canonical way. We work with the simplified form
\begin{equation}
\varepsilon_{(\gamma_1,\delta_1),(\gamma_2,\delta_2)}(\homatinfty{p_2},\cohomatinfty{q_2})=e^{-2\pi i(\langle\homatinfty{p_2},\cohomatinfty{\delta_1}\rangle+\langle\homatinfty{\gamma_1},\cohomatinfty{q_2}\rangle)}
\end{equation}
of the braiding (see eq. \eqref{eq:braidingatinfinity}). Let $\Gamma_1,\Gamma_2$ be two semi-infinite paths on $\mathbb{R}^2$ and $\Delta_1,\Delta_2$ two semi-inifinite dual paths on $\mathbb{R}^2$. If $f_1,f_2$ are $n-1$-cycles on $F$ with coefficients in $\hat{G}$ and $g_1,g_2$ are $n-1$-cocycles on $F$ with coefficients in $G$, then $\gamma_i:=f_i\otimes\Gamma_i$ and $\delta_i:=g_i\otimes\Delta_i$ represent homology and cohomology classes at infinity. These can be interpreted as charges labelled by homology and cohomology classes of $F$. We fix the orientations by requiring that edges in $\Gamma_i$ point towards infinity along the path and edges in $\Delta_i$ point to the left when we look in the direction towards infinity. Assume that $\gamma_i,\delta_i$ are localized in infinite cones $C_i$ such that a counterclockwise rotation takes $C_2$ to $C_1$.

We need to find the classes $\homatinfty{p_2},\cohomatinfty{q_2}$ which mean a counterclockwise rotation of the second charge. Let $P$ be the sum of all $2$-cells in $\mathbb{R}^2$ with clockwise orientation and let $Q$ be the sum of all $0$-cells in $\mathbb{R}^2$ with positive orientation (i.e. the coboundary of a term consists of edges pointing away from the vertex in question). Then $p_2=(-1)^{n-1}f_2\otimes P$ and $q_2=(-1)^{n-1}g_2\otimes Q$ are the relevant (co-)homology classes at infinity. We can calculate the exponent as
\begin{equation}
\begin{split}
\langle\homatinfty{p_2},\cohomatinfty{\delta_1}\rangle
 & = \langle p_2,\partial^T\delta_1\rangle-\langle\partial p_2,\delta_1\rangle  \\
 & = \langle (-1)^{n-1}f_2\otimes P,(-1)^{n-1}g_1\otimes\partial^T\Delta_1\rangle-\langle 0,g_1\otimes\Delta_1\rangle  \\
 & = \langle P,\partial^T\Delta_1\rangle\cdot\langle f_2,g_1\rangle = -\langle f_2,g_1\rangle
\end{split}
\end{equation}
and
\begin{equation}
\begin{split}
\langle\homatinfty{\gamma_1},\cohomatinfty{q_2}\rangle
 & = \langle\gamma_1,\partial^T q_2\rangle-\langle\partial\gamma_1,q_2\rangle  \\
 & = \langle f_1\otimes\Gamma_1,0\rangle-\langle (-1)^{n-1}f_1\otimes\partial\Gamma_1,(-1)^{n-1}g_2\otimes Q\rangle  \\
 & = -\langle \partial\Gamma_1,Q\rangle\cdot\langle f_1,g_2\rangle = -\langle f_1,g_2\rangle,
\end{split}
\end{equation}
therefore the phase in the braiding is
\begin{equation}
e^{-2\pi i(\langle f_2,g_1\rangle+\langle f_1,g_2\rangle)}.
\end{equation}
If we started with the opposite ordering, i.e. a clockwise rotation takes $C_2$ to $C_1$, then we can take $p_2=q_2=0$. In this case the braiding is the identity as expected.

The statistics of these excitations is related to the braiding of an endomorphism with itself. For the endomorphism $\rho_{(\gamma,\delta)}$ with $\gamma:=f\otimes\Gamma$ and $\delta:=g\otimes\Delta$ as before, we either need to take $p=(-1)^{n-1}f\otimes P$ and $q=0$ or $p=0$ and $q=(-1)^{n-1}g\otimes Q$, therefore the phase factor in the twist is $e^{2\pi i\langle f,g\rangle}$.

\section{Cellular homology, homology at infinity}\label{sec:homology}

In this section we collect some facts on CW complexes and various homology groups associated with them. The exposition is based on refs. \cite{Hatcher} and \cite{ends}. For the readers' convenience we recall the definition:
\begin{defn}
A cell complex (or CW complex) $E$ is a topological space defined inductively as follows. $E^0$ is a discrete space, whose points are called 0-cells. For $i\ge 1$ let $\mathcal{E}_i$ be an index set and $\varphi_\alpha:S^{i-1}\to E^{i-1}$ ($\alpha\in \mathcal{E}_i$) a family of maps. Here we regard $S^{i-1}$ as the boundary of the closed $i$-ball $D^i_\alpha$. The $i$-skeleton $E^i$ is formed as $(E^{i-1}\sqcup(\bigsqcup_{\alpha\in \mathcal{E}_i}D^i_\alpha))/\sim$ where $\sim$ is the equivalence relation generated by $\varphi_\alpha(x)\sim x$ whenever $x\in S^{i-1}\simeq\partial D^i_\alpha$. The interior $e_\alpha$ of the ball $D^i_\alpha$ in $E^i$ is also called an $i$-dimensional cell or $i$-cell.

Finally, we let $E$ be the union (more precisely, direct limit) of the spaces $E^i$. If $E^n=E$ for some $n$ (i.e. the index sets $\mathcal{E}_i$ are empty for $i>n$) then the smallest such $n$ is called the dimension of the cell complex.
\end{defn}
For example, a $1$-dimensional cell complex is the same as a graph with $0$-cells as vertices and $1$-cells as edges, while surfaces are special 2-dimensional cell complexes.

For any pair $(\alpha,\beta)\in \mathcal{E}_i\times \mathcal{E}_{i-1}$ we can form the map $S^{i-1}\to E^{i-1}\to S^{i-1}$ which is the composition of the attaching map $\varphi_\alpha$ and the map collapsing the complement of the $i-1$-cell $\beta$ in $E^{i-1}$ to a point. We let $d_{\alpha\beta}$ denote the degree of this map (after choosing an orientation for all the appearing spheres). We will always assume that for each $i$-cell $\beta$ the set $\{\alpha\in \mathcal{E}_{i+1}|d_{\alpha\beta}\neq 0\}$ is finite. This is necessary for some of the appearing sums to make sense.

Next we introduce various chain complexes. To this end we will fix a finite abelian group $G$. We use additive notation for abelian groups, i.e. $0$ is the identity element, $ng=g+g+\cdots+g$ etc. Let $\hat{G}$ denote the dual of $G$, i.e. $\hat{G}=\Hom(G,\mathbb{Q}/\mathbb{Z})$, which may also be identified with the set of characters or equivalence classes of irreducible representations of $G$. Note that the dual of $\hat{G}$ is canonically isomorphic to $G$.

For each $i$ we let
\begin{equation}
C^{\text{lf}}_i(E;\hat{G})=\prod_{\alpha\in \mathcal{E}_i}\hat{G},
\end{equation}
where ``lf'' stands for locally finite. Elements in this group can be thought of as formal infinite linear combinations of $i$-cells with coefficients in $\hat{G}$ and accordingly, we will use the notation $\chi e_\alpha$ when referring to an element of the factor with index $\alpha$. The \emph{locally finite cellular chain complex} of $E$ (with coefficients in $\hat{G}$) is the direct sum
\begin{equation}
C^{\text{lf}}_\bullet(E;\hat{G})=\bigoplus_{i\in\mathbb{N}}C^{\text{lf}}_i(E;\hat{G})
\end{equation}
together with the boundary map $\partial:C^{\text{lf}}_\bullet(E;\hat{G})\to C^{\text{lf}}_\bullet(E;\hat{G})$ which is the homomorphism uniquely determined by
\begin{equation}
\partial\Big(\sum_{\alpha\in \mathcal{E}_i}\chi_\alpha e_\alpha\Big)=\sum_{\beta\in \mathcal{E}_{i-1}}\Big(\sum_{\alpha\in \mathcal{E}_i}d_{\alpha\beta}\chi_\alpha\Big)e_\beta.
\end{equation}

The \emph{cellular chain complex} of $E$ is the subgroup
\begin{equation}
C_\bullet(E;\hat{G})=\bigoplus_{i\in\mathbb{N}}C_i(E;\hat{G})\le C^{\text{lf}}_\bullet(E;\hat{G})
\end{equation}
with
\begin{equation}
C_i(E;\hat{G})=\bigoplus_{\alpha\in \mathcal{E}_i}\hat{G}
\end{equation}
together with the restriction of $\partial$.

Elements of $C^{\text{lf}}_\bullet(E;\hat{G})$ are called locally finite (cellular) chains, a cycle is a chain $b$ such that $\partial b=0$ while chains of the form $\partial d$ are called boundaries. A basic property of the chain complex is that $\partial^2=0$. Thus the group of boundaries is a subgroup of the group of cycles. The quotient group is the locally finite homology. Note that $\partial$ is homogeneous of degree $-1$ (i.e. maps $C^{\text{lf}}_i(E;\hat{G})$ into $C^{\text{lf}}_{i-1}(E;\hat{G})$) and therefore \emph{locally finite homology groups} can be defined for each degree as
\begin{equation}
H^{\text{lf}}_i(E;\hat{G})=(\ker\partial\cap C^{\text{lf}}_i(E;\hat{G}))/(\partial(C^{\text{lf}}_{i+1}(E;\hat{G}))).
\end{equation}
Similarly, the \emph{homology groups} are defnined as
\begin{equation}
H_i(E;\hat{G})=(\ker\partial\cap C_i(E;\hat{G}))/(\partial(C_{i+1}(E;\hat{G}))).
\end{equation}

Given a chain $b=\sum_\alpha \chi_\alpha e_\alpha$ we define its support as $\supp b=\{\alpha\in\mathcal{E}|\chi_\alpha\neq 0\}$. Note that $b$ belongs to $C_\bullet(E;\hat{G})$ iff $|\supp b|<\infty$. It is important to realize that it may happen for some $b\in C^{\text{lf}}_\bullet(E;\hat{G})$ that $\partial b\in C_\bullet(E;\hat{G})$ yet there is no chain $b'\in C_\bullet(E;\hat{G})$ with $\partial b=\partial b'$.

Next we introduce cochain complexes. For each $i$ we let
\begin{equation}
C^i(E;G)=\prod_{\alpha\in \mathcal{E}_i}G.
\end{equation}
Again, we will think of the elements of $C^i(E;G)$ as formal infinite linear combinations of the $e_\alpha$ but this time with coefficients from $G$. As before, we take $C^\bullet(E;G)$ to be the direct sum of the $C^i(E;G)$ for all $i\in\mathbb{N}$. We equip $C^\bullet(E;G)$ with the coboundary map $\partial^T:C^\bullet(E;G)\to C^\bullet(E;G)$ satisfying
\begin{equation}
\partial^T\Big(\sum_{\beta\in \mathcal{E}_i}g_\beta e_\beta\Big)=\sum_{\beta\in \mathcal{E}_i}\sum_{\alpha\in \mathcal{E}_{i+1}}\Big(d_{\alpha\beta}g_\beta\Big)e_\alpha.
\end{equation}

Finally, we will make use of the subgroups
\begin{equation}
C^i_{\text{lf}}(E;G)=\bigoplus_{\alpha\in \mathcal{E}_i}G
\end{equation}
of finite linear combinations of $i$-cells and let $C^\bullet_{\text{lf}}(E;G)$ be their direct sum. With the restriction of $\partial^T$ this is a chain complex called the \emph{locally finite cellular cochain complex}.

We will call an element of $C^\bullet(E;G)$ a cochain, a cochain in $\ker\partial^T$ a cocycle and one in $\im\partial^T$ a coboundary. The support of a cochain $a=\sum_\alpha g_\alpha e_\alpha$ is $\supp a=\{\alpha\in\mathcal{E}|g_\alpha\neq 0\}$. Similarly as before, we define the \emph{locally finite cohomology groups} (with coefficients in $G$) to be
\begin{equation}
H_{\text{lf}}^i(E;G)=(\ker\partial^T\cap C_{\text{lf}}^i(E;G))/(\partial^T(C_{\text{lf}}^{i-1}(E;G)))
\end{equation}
and the \emph{cohomology groups} as
\begin{equation}
H^i(E;G)=(\ker\partial^T\cap C^i(E;G))/(\partial^T(C^{i-1}(E;G))).
\end{equation}

Again, it may happen that $\partial^T a\in C_{\text{lf}}^\bullet$ but there is no $a'\in C_{\text{lf}}^\bullet$ with $\partial^Ta=\partial^Ta'$.

There is a canonical map $C_\bullet(E;\hat{G})\times C^\bullet(E;G)\to\mathbb{Q}/\mathbb{Z}$ defined as
\begin{equation}
\langle\sum_\alpha\chi_\alpha e_\alpha,\sum_\beta g_\beta e_\beta\rangle:=\sum_\alpha\chi_\alpha(g_\alpha)
\end{equation}
and similarly for $C^{\text{lf}}_\bullet(E;\hat{G})$ and $C_{\text{lf}}^\bullet(E;G)$. In both cases the appearing sum contains finitely many terms, but in general, it is not possible to extend these to a map $C^{\text{lf}}_\bullet(E;\hat{G})\times C^\bullet(E;G)\to\mathbb{Q}/\mathbb{Z}$ even though the two maps agree on $C_\bullet(E;\hat{G})\times C_{\text{lf}}^\bullet(E;G)$. The boundary and coboundary maps satisfy $\langle b,\partial^T a\rangle=\langle\partial b,a\rangle$ whenever at most one of $a$ and $b$ has infinite support. However, it may happen that $|\supp a|=|\supp b|=\infty$, $|\supp\partial^Ta|<\infty$ and $|\supp\partial b|<\infty$ (in particular, both sides make sense) but $\langle b,\partial^T a\rangle\neq\langle\partial b,a\rangle$. As a simple example, let $E=\mathbb{R}$ with $0$-cells at each integer point and $1$-cells joining $i$ and $i+1$ for $i\in\mathbb{Z}$, take $G=\hat{G}=\mathbb{Z}_2$ and let $a$ be the sum of points corresponding to the nonnegative integers and let $b$ be the sum of all the $1$-cells. Then $\partial^Ta$ is the $1$-cell joining $-1$ and $0$, while $\partial b=0$, therefore $\langle b,\partial^T a\rangle=1/2$ but $\langle\partial b,a\rangle=0$.

Recall that $C_\bullet(E;\hat{G})$ is regarded as a subspace of $C^{\text{lf}}_\bullet(E;\hat{G})$ equipped with the restricion of the boundary map. In other words the inclusion $i:C_\bullet(E;\hat{G})\to C^{\text{lf}}_\bullet(E;\hat{G})$ is a chain map, i.e. $\partial\circ i=i\circ\partial$. In general, the homology and locally finite homology groups (which are homology groups of the respective chain complexes) are different and accordingly, the map $i_*:H_\bullet(E;\hat{G})\to H^{\text{lf}}_\bullet(E;\hat{G})$ induced by $i$ fails to be an isomorphism.

In homological algebra, the standard tools to measure the extent to which a chain map fails to induce isomorphisms on the homology groups are the algebraic mapping cone and the associated long exact sequence (see e.g. in \cite{HiltonStammbach}). In general, for a chain map $f:A_\bullet\to B_\bullet$ we define its algebraic mapping cone to be the graded abelian group $C_\bullet=A_{\bullet-1}\oplus B_{\bullet}$ where the $-1$ refers to a shift in the grading (i.e. $C_n=A_{n-1}\oplus B_n$) equipped with the boundary map $\partial(a,b)=(-\partial a,f(a)+\partial b)$.

The maps $r:B_\bullet\to C_\bullet$ defined as $b\mapsto(0,b)$ and $s:C_\bullet\to A_{\bullet-1}$ given by $(a,b)\mapsto -a$ are chain maps forming a short exact sequence
\begin{equation}
\xymatrix{0\ar[r] & B_\bullet \ar[r]^r & C_\bullet \ar[r]^s & A_{\bullet-1} \ar[r] & 0}
\end{equation}
while the corresponding long exact sequence for homology groups reads
\begin{equation}
\xymatrix{\cdots\ar[r]^{s_*} & H_n(A) \ar[r]^{f_*} & H_n(B) \ar[r]^{r_*} & H_n(C) \ar[r]^{s_*} & H_{n-1}(A) \ar[r]^{f_*} & \cdots}
\end{equation}
i.e. the connecting homomorphism is equal to $f_*$.

Applying this construcion to the inclusion $i:C_\bullet(E;\hat{G})\to C^{\text{lf}}_\bullet(E;\hat{G})$ results in the long exact sequence
\begin{equation}\label{eq:longhomology}
\xymatrix{\cdots H_n(E;\hat{G}) \ar[r]^{i_*} & H^{\text{lf}}_n(E;\hat{G}) \ar[r]^{r_*} & H^{\infty}_{n-1}(E;\hat{G}) \ar[r]^{s_*} & H_{n-1}(E;\hat{G})\cdots},
\end{equation}
where the groups $H^{\infty}_n(E;\hat{G})$ are defined to be the homology groups of the algebraic mapping cone of $i$ up to a shift in the grading:
\begin{equation}
H^{\infty}_n(E;\hat{G})=H_n(C_{\bullet+1})
\end{equation}
with $C$ similarly defined as above. These groups constitute the \emph{homology of $E$ at infinity}.

Injectivity of the chain map $i$ implies that the homology at infinity has a simpler description. Namely, regarding $C_\bullet(E;\hat{G})$ as a subspace of $C^{\text{lf}}_\bullet(E;\hat{G})$ as before we can write $\partial(b,d)=(-\partial b,b+\partial d)$, therefore a pair $(b,d)$ is a cycle iff $b=-\partial d$. This means that an $n$-cycle at infinity is essentially a locally finite $(n+1)$-chain with finite boundary, while an $n$-boundary at infinty is a locally finite $(n+1)$-boundary up to a term of the form $(-\partial b,b)$ with $|\supp b|<\infty$. We will denote by $\homatinfty{b}$ the homology class at infinity represented by $(-\partial b,b)$ when $|\supp \partial b|<\infty$.

Similarly, one can introduce the \emph{cohomology groups of $E$ at infinity}, which fit into the long exact sequence
\begin{equation}\label{eq:longcohomology}
\xymatrix{\cdots H_{\text{lf}}^n(E;G) \ar[r]^{i^*} & H^n(E;G) \ar[r] & H_{\infty}^n(E;G) \ar[r] & H_{\text{lf}}^{n+1}(E;G) \cdots}
\end{equation}
where the connecting homomorphism is equal to $i^*$.

Again, since $C_{\text{lf}}^\bullet(E;G)\le C^\bullet(E;G)$, we can write $\partial^T(a,c)=(-\partial^Ta,a+\partial^Tc)$, therefore an $n$-cocycle at infinity is essentially an $n$-cochain with finite coboundary, and an $n$-boundary at infinity is an $n$-boundary up to a term of the form $(-\partial^T a,a)$ with $|\supp a|<\infty$. The cohomology class at infinity represented by $(-\partial^T a,a)$ will be denoted by $\cohomatinfty{a}$ when $|\supp\partial^T a|<\infty$.

For sufficiently nice spaces there is a more topological way to view homology and cohomology at infinity. Namely, if we define the \emph{end space} $e(E)$ as the set of proper maps $[0,\infty)\to E$ with the compact-open topology then the singular (co-)homology groups of $e(E)$ are the same as (co-)homology groups of $E$ at infinity. For example $\dim H^{\infty}_0(E,\mathbb{Q})$ is the number of path components of $e(E)$, each of which corresponds to a possible ``direction'' in which a moving point can escape to infinity.

It is important to note that while homotopic maps induce the same homomorphisms on homology and cohomology groups, this is not the case for the locally finite variants and (co-)homology groups at infinity. In fact, not every continuous map induces such homomorphisms. Instead, we need to restrict to \emph{proper} maps, i.e. those under which the preimage of any compact set is compact, and \emph{proper homotopies}, i.e. homotopies between proper maps which are themselves proper. As a consequence, homotopy equivalent spaces may have non-isomorphic homology groups at infinity, but the proper homotopy type of the space does determine locally finite (co-)homology groups as well as (co-)homology groups at infinity.

Finally, we introduce a pairing between homology and cohomology groups at infinity. Let $\homatinfty{d}\in H^\infty_n(E;\hat{G})$ and let $\cohomatinfty{a}\in H_\infty^n(E;G)$. This means that $d\in C^{\text{lf}}_{n+1}(E;\hat{G})$, $a\in C^n(E;G)$ and $|\supp\partial d|<\infty$, $|\supp\partial^T a|<\infty$. Then we can form the difference
\begin{equation}
\langle\homatinfty{d},\cohomatinfty{a}\rangle:=\langle d,\partial^T a\rangle-\langle \partial d,a\rangle\in\mathbb{Q}/\mathbb{Z}.
\end{equation}
In order to see that this does not depend on the representatives, it is enough to prove that if either $(-\partial d,d)$ is a boundary at infinity or $(-\partial^T a,a)$ is a coboundary at infinity, the difference is $0$. Indeed, if $d=\partial f+d'$ with $d'\in C_{n+1}(E;\hat{G})$ then
\begin{equation}
\begin{split}
\langle d,\partial^Ta\rangle-\langle \partial d,a\rangle
 & = \langle\partial f+d',\partial^Ta\rangle-\langle \partial (\partial f+d'),a\rangle  \\
 & = \langle\partial f,\partial^Ta\rangle+\langle d',\partial^Ta\rangle-\langle\partial^2f,a\rangle-\langle\partial d',a\rangle  \\
 & = \langle f,(\partial^T)^2a\rangle+\langle d',\partial^Ta\rangle-\langle\partial^2f,a\rangle-\langle d',\partial^T a\rangle  \\
 & = 0
\end{split}
\end{equation}
since $|\supp\partial^T a|<\infty$ and $|\supp d'|<\infty$. The other case is similar. Therefore we get a well defined bilinear map $H^\infty_n(E;\hat{G})\times H_\infty^n(E;G)\to\mathbb{Q}/\mathbb{Z}$.

\section{Kitaev model on cell complexes}\label{sec:model}

In this section we will introduce the model we are working with. Since we are mainly interested in systems with infinitely many sites, we use the C*-algebraic framework, which we briefly recapitulate (see ref. \cite{BR1}). The basic objects in our context are quasi-local algebras and interactions. For simplicity the local Hilbert spaces will be the same and finite dimensional. Let $S$ be a set and $\mathcal{H}$ a Hilbert space with $\dim\mathcal{H}\le\infty$. For each finite subset $\Lambda\subseteq S$ we let $\mathcal{H}_\Lambda\simeq\mathcal{H}^{\otimes|\Lambda|}$ and when $\Lambda\cap\Lambda'=\emptyset$ we identify $\mathcal{H}_{\Lambda\cup\Lambda'}$ with $\mathcal{H}_\Lambda\otimes\mathcal{H}_{\Lambda'}$. We also let $\mathfrak{A}_\Lambda=\boundeds(\mathcal{H}_\Lambda)$ be the C*-algebra of bounded operators acting on the corresponding Hilbert space.

In this way for each pair $\Lambda_1\subseteq\Lambda_2\subseteq S$ of finite subsets one can define an inclusion $\mathfrak{A}_{\Lambda_1}\hookrightarrow\mathfrak{A}_{\Lambda_2}$ by $A\mapsto A\otimes\id_{\mathcal{H}_{\Lambda_2\setminus\Lambda_1}}$. These maps form a direct system (in the category of C*-algebras), the direct limit of which is called the quasi-local algebra and will be denoted $\mathfrak{A}$. Its elements are called (quasi-local) observables.

This object comes equipped with injective homomorphisms $\mathfrak{A}_\Lambda\hookrightarrow\mathfrak{A}$ and we will henceforth identify the algebras $\mathfrak{A}_\Lambda$ with their images in $\mathfrak{A}$. It is then true that the union of the subalgebras $\mathfrak{A}_\Lambda$ for all finite subsets $\Lambda$ is dense in $\mathfrak{A}$. Elements in this dense subset (in fact, subalgebra) are called local, and the support of a local observable $A$ (denoted by $\supp A$) is the smallest $\Lambda$ (with respect to set containment) such that $A\in\mathfrak{A}_\Lambda$.

Similarly, for any subset $S'\subseteq S$ we define the subalgebra $\mathfrak{A}_{S'}$ as the closure of the union of the subalgebras $\mathfrak{A}_\Lambda$ with $\Lambda\subseteq S'$ finite.

Later we will encounter families of points in topological spaces indexed by finite subsets of an infinite set. Since the set of finite subsets is a directed set ordered by inclusion, these families are actually nets. We will use the notation $\lim_{\Lambda\to S}$ to refer to limits of such nets. We will use the fact that if the topological space in question is a complete metric space then Cauchy nets converge.

In the case of infinite systems it is no longer possible to specify a general time evolution in terms of a Hamiltonian generating a 1-parameter group of unitaries, but instead we need to deal with a 1-parameter group of automorphisms of $\mathfrak{A}$. One way to construct such evolutions is to consider local Hamiltonians defined in terms of interactions as follows. To each finite subset $\Lambda\subseteq S$ we take a self-adjoint element $\Phi_\Lambda$ in $\mathfrak{A}_\Lambda$ then let
\begin{equation}
H_\Lambda=\sum_{\Lambda'\subseteq\Lambda}\Phi_{\Lambda'}
\end{equation}
and for $A\in\mathfrak{A}_\Lambda$ we define
\begin{equation}
\delta(A)=\lim_{\Lambda\to S}i[H_\Lambda,A],
\end{equation}
where the limit is taken over the set of finite subsets of $S$ ordered by inclusion. Under certain conditions on the interaction terms $\Phi_\Lambda$, the closure of this operator generates an action $\alpha_t$ of $\mathbb{R}$ on $\mathfrak{A}$.

Recall that Kitaev's toric code model is defined on qubits living on the edges (1-cells) of a graph drawn on a surface. Generalising this idea, we consider cells of arbitrary dimension instead of 1-cells. Correspondingly, the role of the surface will be played by a cell complex.

In what follows we fix a finite abelian group $G$ and a finite dimensional cell complex $E$ satisfying the assumption above. The index set parametrizing the cells of dimension $i$ will be denoted by $\mathcal{E}_i$ and we let $\mathcal{E}=\bigsqcup_{i\in\mathbb{N}}\mathcal{E}_i$ be their disjoint union. This will be the set of sites in our system. The Hilbert space at each site will be $\mathcal{H}=\ell^2(G)$ (with orthonormal basis $\{\ket{g}\}_{g\in G}$). We let $\mathfrak{A}$ be the corresponding quasi-local algebra. For $A\in\boundeds(\mathcal{H})$ we use the notation $(A)_\alpha$ to refer to the corresponding operator acting at the site $\alpha$.

The interaction will be given in terms of generalised Pauli matrices. We let these act on $\ket{h}\in\mathcal{H}$ as
\begin{equation}
X^g\ket{h}=\ket{g+h}\text{ and }Z^\chi\ket{h}=\chi(h)\ket{h},
\end{equation}
where $g\in G$ and $\chi:G\to\mathbb{C}$ is an irreducible character. Note that both types of operators are unitary. Next we take tensor products of these to define local operators corresponding to chains and cochains. For $b=\sum_\alpha \chi_\alpha e_\alpha\in C_\bullet(E;\hat{G})$ we let
\begin{equation}
Z^b=\prod_{\alpha\in\mathcal{E}}(Z^{\chi_\alpha})_\alpha
\end{equation}
and for $a=\sum_\alpha g_\alpha e_\alpha\in C_{\text{lf}}^\bullet(E;G)$ we let
\begin{equation}
X^a=\prod_{\alpha\in\mathcal{E}}(X^{g_\alpha})_\alpha.
\end{equation}
Note that since $X^0=Z^0=I$ (with $0$ the neutral element and the trivial character) in the tensor products above only finitely many non-identity factors are present and hence we get well-defined elements in the local algebra. In fact, $\supp X^a=\supp a$ and $\supp Z^b=\supp b$.

It is easy to see that scalar multiples of operators of the form $X^aZ^b$ form a group and that both $a\mapsto X^a$ and $b\mapsto Z^b$ are group homomorphisms. Moreover, they satisfy the following commutation relations:
\begin{subequations}\label{eq:XZcommutator}
\begin{align}
X^aX^{a'} & =X^{a'}X^a=X^{a+a'}  \\
Z^bZ^{b'} & =Z^{b'}Z^b=Z^{b+b'}  \\
X^aZ^b & =e^{-2\pi i\langle b,a\rangle} Z^bX^a.
\end{align}
\end{subequations}

Finally, the interaction terms will be higher dimensional analogues of the star and plaquette-operators (see Figure \ref{fig:starplaq}). For each $\alpha,\beta\in\mathcal{E}$ we let
\begin{equation}\label{eq:starplaquette}
A_\alpha=\frac{1}{|G|}\sum_{g\in G}X^{\partial^T(g e_\alpha)}\text{ and }B_\beta=\frac{1}{|G|}\sum_{\chi\in \hat{G}}Z^{\partial(\chi e_\beta)}
\end{equation}
and for a finite subset $\Lambda\subseteq\mathcal{E}$ we let
\begin{equation}
H_\Lambda=-\sum_{\substack{\alpha\in\mathcal{E}  \\  \supp A_\alpha\subseteq\Lambda}}A_\alpha-\sum_{\substack{\beta\in\mathcal{E}  \\  \supp B_\beta\subseteq\Lambda}}B_\beta
\end{equation}
be the local Hamiltonian.
\begin{figure}
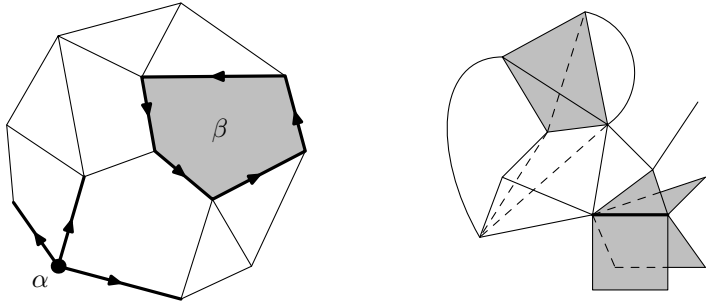

\begin{subfigure}{.45\textwidth}
\centering
\includegraphics{starplaquette1.mps}
\end{subfigure}
\begin{subfigure}{.45\textwidth}
\centering
\includegraphics{starplaquette2.mps}
\end{subfigure}
\caption{Left: A CW structure on the disk. Supports of the star operator $A_\alpha$ and the plaquette operator $B_\beta$ are marked with arrows. Here $n=1$, $\alpha$ is a $0$-chain (vertex) and $\beta$ is a $2$-chain (face). Right: A three dimensional CW complex. Supports of a generalised star operator (bottom) and a generalised plaquette operator (top) highlighted in the $n=2$ case.\label{fig:starplaq}}
\end{figure}

It follows from the commutation relations that the terms appearing in the Hamiltonian commute with each other. Indeed, for any $c\in C_{\text{lf}}^\bullet(E;G)$ and $d\in C_\bullet(E;\hat{G})$ we have
\begin{equation}
X^{\partial^Tc}Z^{\partial d}=e^{-2\pi i\langle \partial d,\partial^T c\rangle}Z^{\partial d}X^{\partial^Tc}=e^{-2\pi i\langle \partial^2 d,c\rangle}Z^{\partial d}X^{\partial^Tc}=Z^{\partial d}X^{\partial^Tc},
\end{equation}
therefore each term in $A_\alpha$ commutes with each term in $B_\beta$.

When $G=\mathbb{Z}_d$ our $X$ and $Z$ operators specialize to the usual qudit Pauli matrices and tensor products of their powers. In the $d=2$ case these are proportional to the spin observables of a spin-$\frac{1}{2}$ particle.

Note that while a derivation $\delta:\mathfrak{A}\to\mathfrak{A}$ as above may always be defined using our interaction, it may or may not generate a time evolution depending on the structure of $E$. As an important example, if both $|\supp\partial(\chi e_\alpha)|$ and $|\supp\partial^T(ge_\alpha)|$ is bounded then a time evolution exists \cite[Proposition 6.2.3]{BR2}.

The set $\mathcal{E}$ of sites is partitioned into subsets according to the dimensions of the cells in the cell complex. The definition makes it clear that the local Hamiltonians do not introduce any interaction between cells of different dimension. This implies that our system is in fact composed of independent subsystems, and hence it is enough to analyze them separately. Mathematically, $\mathfrak{A}$ is the tensor product\footnote{The algebras $\mathfrak{A}_{\mathcal{E}_i}$ are nuclear (being direct limits of finite dimensional algebras), hence they have a unique tensor product.} of the $\mathfrak{A}_{\mathcal{E}_i}$ and the local Hamiltonians $H_\Lambda$ with $\Lambda\subseteq\mathcal{E}_i$ generate separate commuting time evolutions $\alpha^{(i)}_t$ such that $\alpha_t=\alpha^{(0)}_t\circ\cdots\circ\alpha^{(\dim E)}_t$.

For this reason we will single out one of the subsystems and modify the model accordingly. This means that our model now depends on a cell complex $E$, a finite abelian group $G$ and a number $n\in\mathbb{N}$, the C*-algebra is the quasi-local algebra $\mathfrak{A}:=\mathfrak{A}_{\mathcal{E}_n}$ and the time evolution is generated by the local Hamiltonians
\begin{equation}\label{eq:localHamiltonian}
H_\Lambda=-\sum_{\substack{\alpha\in\mathcal{E}_{n-1}  \\  \supp A_\alpha\subseteq\Lambda}}A_\alpha-\sum_{\substack{\beta\in\mathcal{E}_{n+1}  \\  \supp B_\beta\subseteq\Lambda}}B_\beta.
\end{equation}

\section{Ground states}\label{sec:ground}

In this section we determine the set of frustration free ground states. In general, there exist ground states which do not have this property, but we do not have a complete classification of them. In a recent preprint by Cha, Naaijkens and Nachtergaele \cite{CNN}, all ground states of the planar model have been classified. We believe that their method can be adapted to our more general model as well.

Recall that a state on a C*-algebra $\mathfrak{A}$ is a linear functional $\omega:\mathfrak{A}\to\mathbb{C}$ which is positive, i.e. $\omega(A^*A)\ge 0$ for all $A\in\mathfrak{A}$ and normalized, i.e. $\norm{\omega}=1$. If $\mathfrak{A}$ is unital (as in our case), the second condition may be replaced with $\omega(I)=1$ where $I$ is the unit.

Given a time evolution $\alpha_t$ with generator $\delta$, ground states are defined to be states such that $-i\omega(A^*\delta(A))\ge0$ for all finite $\Lambda$ and $A\in\mathfrak{A}_\Lambda$ (and hence for any $A$ in the domain of $\delta$) \cite{BR2}.

We will make use of the following lemma:
\begin{lem}\label{lem:unitary}
Let $\mathfrak{A}$ be a unital C*-algebra, $\omega:\mathfrak{A}\to\mathbb{C}$ a state and $U\in\mathfrak{A}$ a unitary element, i.e. $UU^*=U^*U=I$. Then for any $Y\in\mathfrak{A}$ the inequalities
\begin{equation}
\left|\omega(U)\omega(Y)-\omega(UY)\right|\le\sqrt{1-|\omega(U)|^2}\sqrt{\omega(Y^*Y)}
\end{equation}
and
\begin{equation}
\left|\omega(Y)\omega(U)-\omega(YU)\right|\le\sqrt{1-|\omega(U)|^2}\sqrt{\omega(Y^*Y)}
\end{equation}
hold. In particular, if $|\omega(U)|=1$ then $\omega(U)\omega(Y)=\omega(UY)=\omega(YU)$.
\end{lem}
\begin{proof}
By the Cauchy-Schwarz inequality we have
\begin{equation}
\begin{split}
\left|\omega(U)\omega(Y)-\omega(UY)\right|^2
 & = \left|\omega((\omega(U)I-U)Y)\right|^2  \\
 & \le \omega((\omega(U)I-U)^*(\omega(U)I-U))\omega(Y^*Y)  \\
 & = \omega(|\omega(U)|^2I+I-\omega(U)U^*-\overline{\omega(U)}U)\omega(Y^*Y)  \\
 & = (1-|\omega(U)|^2)\omega(Y^*Y).
\end{split}
\end{equation}
The proof of the other inequality is similar.
\end{proof}
This lemma plays a similar role as Lemma 2.1 in ref. \cite{NaaijkensLocalized} (proved in \cite{AFH}). In the $G=\mathbb{Z}_2$ case either lemma can be used since the qubit $X$ and $Z$ operators are both unitary and selfadjoint.

Consider the operators $X^{\partial^Tc}$ and $Z^{\partial d}$ for each $c\in C_{\text{lf}}^{n-1}(E;G)$ and $d\in C_{n+1}(E;\hat{G})$. As noted earlier, these operators commute and the interaction terms are linear combinations of them. Moreover, such operators form a group under multiplication, therefore the closure of the subspace spanned by them is a C*-subalgebra $\mathfrak{A}_{XZ}$ of $\mathfrak{A}$. Let $\omega:\mathfrak{A}_{XZ}\to\mathbb{C}$ be the linear functional determined by $\omega(X^{\partial^Tc}Z^{\partial d})=1$. Here existence is guaranteed by the fact that such operators are linearly independent. It is easy to verify that $\omega$ is in fact a state.\footnote{Use the fact that $(X^{\partial^Ta_1}Z^{\partial b_1})^*(X^{\partial^Ta_2}Z^{\partial b_2})=X^{\partial^T(a_2-a_1)}Z^{\partial (b_2-b_1)}$.} Clearly, a frustration free ground state is a state whose restriction to $\mathfrak{A}_{XZ}$ is $\omega$.

$\omega$ can be extended to a state $\omega_0$ on $\mathfrak{A}$. Using the lemma above one can see that for any local observable $A\in\mathfrak{A}_\Lambda$
\begin{equation}
\begin{split}
-i\omega_0(A^*\delta(A))
 & = \sum_{\alpha\in\mathcal{E}_{n-1}}\frac{1}{|G|}\sum_{g\in G}\left(\omega_0(A^*AX^{\partial^T(g e_\alpha)})-\omega_0(A^*X^{\partial^T(g e_\alpha)}A)\right)  \\
 & \qquad+\sum_{\beta\in\mathcal{E}_{n+1}}\frac{1}{|G|}\sum_{\chi\in \hat{G}}\left(\omega_0(A^*AZ^{\partial(\chi e_\beta)})-\omega_0(A^*Z^{\partial(\chi e_\beta)}A)\right)  \\
 & = \sum_{\alpha\in\mathcal{E}_{n-1}}\left(\omega_0(A^*A)-\omega_0\left(A^*\frac{1}{|G|}\sum_{g\in G}X^{\partial^T(g e_\alpha)}A\right)\right)  \\
 & \qquad+\sum_{\beta\in\mathcal{E}_{n+1}}\left(\omega_0(A^*A)-\omega_0\left(A^*\frac{1}{|G|}\sum_{\chi\in \hat{G}}Z^{\partial(\chi e_\beta)}A\right)\right) \ge 0
\end{split}
\end{equation}
since for any projection $P$ we have $\omega_0(A^*(I-P)A)\ge 0$. Therefore any extension $\omega_0$ is a ground state.

We would like to see how much freedom is there in the choice of $\omega_0$. Let $a\in C_{\text{lf}}^n(E;G)$ and $b\in C_n(E;\hat{G})$ and consider the product $X^aZ^b$. Such observables form a basis of the algebra of local observables, which is dense in $\mathfrak{A}$, therefore $\omega_0$ is completely determined by the values $\omega_0(X^aZ^b)$. First note that for any $c\in C_{\text{lf}}^{n-1}(E;G)$
\begin{equation}
\begin{split}
\omega_0(X^aZ^b)
 & = \omega_0(X^{\partial^T c}X^aZ^bX^{-\partial^T c})  \\
 & = e^{-2\pi i\langle b,\partial^T c\rangle}\omega_0(X^{\partial^T c}X^aX^{-\partial^T c}Z^b)  \\
 & = e^{-2\pi i\langle \partial b,c\rangle}\omega_0(X^aZ^b)
\end{split}
\end{equation}
by Lemma \ref{lem:unitary} and similarly, for any $d\in C_{n+1}(E;\hat{G})$
\begin{equation}
\begin{split}
\omega_0(X^aZ^b)
 & = \omega_0(Z^{\partial d}X^aZ^bZ^{-\partial d})  \\
 & = e^{2\pi i\langle \partial d,a\rangle}\omega_0(X^aZ^{\partial d}Z^bZ^{-\partial d})  \\
 & = e^{2\pi i\langle d,\partial^T a\rangle}\omega_0(X^aZ^b),
\end{split}
\end{equation}
therefore $\omega_0(X^aZ^b)=0$ unless $\partial^Ta=0$ and $\partial b=0$.

The closure of the linear span of the elements $X^aZ^b$ with $\partial^Ta=0$ and $\partial b=0$ is the same as the commutant of $\mathfrak{A}_{XZ}$. Frustration free ground states are therefore in one-to-one correspondence with those states on $\mathfrak{A}_{XZ}'$ which restrict to $\omega$ on $\mathfrak{A}_{XZ}$.

Let $c\in C_{\text{lf}}^{n-1}(E;G)$ and $d\in C_{n+1}(E;\hat{G})$. Then
\begin{equation}
\omega_0(X^aZ^b)=\omega_0(X^{\partial^Tc}X^aZ^bZ^{\partial d})=\omega_0(X^{a+\partial^Tc}Z^{b+\partial d})
\end{equation}
using Lemma \ref{lem:unitary} again. Therefore when $\partial^Ta=0$ and $\partial b=0$ then $\omega_0(X^aZ^b)$ depends only on the homology class represented by $b$ and the locally finite cohomology class represented by $a$. Let $J$ denote the closure of the subspace spanned by elements of the form $X^aZ^b-X^{a+\partial^Tc}Z^{b+\partial d}$ where $\partial^Ta=0$ and $\partial b=0$. Then $J$ is a closed two-sided ideal in $\mathfrak{A}_{XZ}'$. Let $p:\mathfrak{A}_{XZ}'\to\mathfrak{A}_{XZ}'/J$ be the quotient map.

If $\tilde{\omega}:\mathfrak{A}_{XZ}'/J\to\mathbb{C}$ is a state then $\tilde{\omega}\circ p$ is a state on $\mathfrak{A}_{XZ}'$. Moreover, since $I-X^{\partial^Tc}Z^{\partial d}\in J$ we have $1=(\tilde{\omega}\circ p)(I)=(\tilde{\omega}\circ p)(X^{\partial^Tc}Z^{\partial d})$, i.e. $\tilde{\omega}\circ p$ extends to a frustration free ground state $\omega_0$. Conversely, if $\omega_0$ restricts to $\omega$ on $\mathfrak{A}_{XZ}$ then it is identically $0$ on $J$ and therefore determines a state $\tilde{\omega}$ on $\mathfrak{A}_{XZ}'/J$. It is not difficult to see that the two constructions are inverses of each other and hence give a bijection between the set of frustration free ground states on $\mathfrak{A}$ and the set of states on $\mathfrak{A}_{XZ}'/J$.

In the usual CSS code setting $\mathfrak{A}$ is finite dimensional, therefore every ground state is frustration free and these correspond to vectors in the code subspace and their mixtures. These vectors represent the possible pure logical states encoded in a larger Hilbert space. By analogy, we may interpret states on $\mathfrak{A}_{XZ}'/J$ as the logical states and $\mathfrak{A}_{\mathrm{logical}}:=\mathfrak{A}_{XZ}'/J$ as the algebra of logical observables.

Under the bijection above pure states correspond to pure ones:
\begin{prop}\label{prop:purestrongground}
A state $\tilde{\omega}:\mathfrak{A}_{\mathrm{logical}}\to\mathbb{C}$ is pure iff the extension $\omega_0:\mathfrak{A}\to\mathbb{C}$ of $\tilde{\omega}\circ p:\mathfrak{A}_{XZ}'\to\mathbb{C}$ is pure.
\end{prop}
\begin{proof}
If $\tilde{\omega}$ is mixed then it can be written as $\lambda\tilde{\omega}_1+(1-\lambda)\tilde{\omega}_2$ for some $\lambda\in(0,1)$ and states $\tilde{\omega}_1\neq\tilde{\omega}_2$. Composing both sides with $p$ from the right and then taking the unique extensions to $\mathfrak{A}$ gives a nontrivial convex decomposition of $\omega_0$.

In the other direction, if $\omega_0$ is mixed then it can be written as $\lambda\omega_1+(1-\lambda)\omega_2$ for some $\lambda\in(0,1)$ and states $\omega_1\neq\omega_2$. We claim that these states must also be frustration free. Indeed, since $X^{\partial^Tc}Z^{\partial d}$ is unitary for any $c$, $d$, the inequality $|\omega_i(X^{\partial^Tc}Z^{\partial d})|\le 1$ must hold for $i=1,2$. But
\begin{equation}
(\lambda\omega_1+(1-\lambda)\omega_2)(X^{\partial^Tc}Z^{\partial d})=\omega(X^{\partial^Tc}Z^{\partial d})=1,
\end{equation}
and this is only possible if $\omega_1(X^{\partial^Tc}Z^{\partial d})=\omega_2(X^{\partial^Tc}Z^{\partial d})=1$.
\end{proof}

To shed more light on the structure of the logical C*-algebra $\mathfrak{A}_{\mathrm{logical}}$, let us note first that the products $X^aZ^b$ with $\partial^Ta=0$, $\partial b=0$ form a basis of a dense subalgebra in $\mathfrak{A}_{XZ}$. By the definition of $J$, it follows that two such products $X^aZ^b$ and $X^{a'}Z^{b'}$ are mapped to the same element in the quotient iff $a-a'$ is a boundary in $C_n(E;\hat{G})$ and $b-b'$ is a locally finite coboundary in $C_{\text{lf}}^n(E;G)$, i.e. iff $a$ and $a'$ represent the same homology class $[a]$ and similarly, $b$ and $b'$ represent the same locally finite cohomology class $[b]$. Let us denote the corresponding element in $\mathfrak{A}_{\mathrm{logical}}$ by $X^{[a]}Z^{[b]}$. Clearly, the set
\begin{equation}
\{X^{[a]}Z^{[b]}|[a]\in H_n(E;\hat{G})\text{ and }[b]\in H_{\text{lf}}^n(E;G)\}
\end{equation}
is linearly independent. Scalar multiples of its elements form a group under multiplication and generate a dense subalgebra. It is easy to work out the multiplication table using eq. \eqref{eq:XZcommutator}:
\begin{equation}
(X^{[a]}Z^{[b]})(X^{[a']}Z^{[b']})=e^{2\pi i\langle b,a'\rangle}X^{[a]+[a']}Z^{[b]+[b']},
\end{equation}
which is well defined since the pairing between $C_n(E;\hat{G})$ and $C_{\text{lf}}^n(E;G)$ descends to a bilinear map $H_n(E;\hat{G})\times H_{\text{lf}}^n(E;G)\to\mathbb{Q}/\mathbb{Z}$ (e.g. $\langle b,\partial^T c\rangle=\langle\partial b,c\rangle=0$ if $b$ is a cycle).

Suppose that $K\le H_n(E;\hat{G})\times H_{\text{lf}}^n(E;G)$ is a subgroup. Then the closure of the linear span of $\{X^{[a]}Z^{[b]}|([a],[b])\in K\}$ is a C*-subalgebra of dimension $|K|$. Since the order of any element in $H_n(E;\hat{G})$ and $H_{\text{lf}}^n(E;G)$ is finite, their product can be written as the union of its finite subgroups. In this way we can find a net of subalgebras of $\mathfrak{A}_{\mathrm{logical}}$ with dense union. One such net is indexed by finite subgroups $K$ as above, and therefore $\mathfrak{A}_{\mathrm{logical}}$ can be identified with the direct limit of this net \cite[Proposition 11.4.1]{KadisonRingrose}.

These properties make it possible to realize $\mathfrak{A}_{\mathrm{logical}}$ as a variant of the CCR algebra without referring to $\mathfrak{A}$, thus we have the following:
\begin{prop}\label{prop:logicalalgebra}
Let $H=H_n(E;\hat{G})\times H_{\text{lf}}^n(E;G)$ and equip it with the bilinear form $(([a],[b]),([a'],[b'])):=\langle b,a'\rangle$. To each finite subgroup $K\le H$ we can associate a finite dimensional C*-algebra acting on the Hilbert space $\mathcal{H}_K:=\ell^2(K)$ with $\{\ket{k}\}_{k\in K}$ an orthonormal basis. Namely, for each $k\in K$ we introduce the unitary $W_K(k)$ acting as $W_K(k)\ket{k'}=e^{2\pi i(k,k')}\ket{k+k'}$ and our C*-algebra is the subalgebra $\mathfrak{A}_K$ of $\boundeds(\mathcal{H}_K)$ generated by such operators. When $K\le K'\le H$ are two subgroups, there is an injective unital homomorphism $\mathfrak{A}_K\to\mathfrak{A}_{K'}$ defined by $W_K(k)\mapsto W_{K'}(k)$. These form a directed system of C*-algebras and its inductive limit is isomorphic to $\mathfrak{A}_{\mathrm{logical}}$. An isomorphism is induced by the maps $W_K(([a],[b]))\mapsto X^{[a]}Z^{[b]}$.
\end{prop}

This description is particularly transparent when $H_n(E;\hat{G})\times H_{\text{lf}}^n(E;G)$ is finite and therefore the algebra $\mathfrak{A}_{\mathrm{logical}}$ is finite dimensional. In this case it is enough to take $K=H$ to get the direct limit, and the algebra itself becomes a finite dimensional (generalized) CCR algebra. The following lemma gives the structure of such algebras in a more general setting.
\begin{lem}\label{thm:ccr}
Let $H$ be a finite group, $(\cdot,\cdot):H\times H\to\mathbb{Q}/\mathbb{Z}$ any bilinear map. For each $h\in H$ let $W(h)$ be the operator acting on $\mathcal{H}=\ell^2(H)$ as $W(h)\ket{h'}=e^{2\pi i(h,h')}\ket{h+h'}$ and let $\mathfrak{A}\le\boundeds(\mathcal{H})$ be the C*-algebra generated by these operators. Let $H_0=\{h\in H|\forall h'\in H:(h,h')-(h',h)=0\}$, $c=\log_2|H_0|$ and $q=\frac{1}{2}\log_2\frac{|H|}{|H_0|}$. Then
\begin{equation}
\mathfrak{A}\simeq\mathbb{C}^{2^c}\otimes\boundeds(\mathbb{C}^{2^q}),
\end{equation}
where the first tensor factor is regarded as a commutative C*-algebra with the pointwise operations.
\end{lem}
\begin{proof}
Consider the operators $U(h)\ket{h'}=e^{2\pi i(h',h)}\ket{h+h'}$. Then
\begin{align}
\begin{split}
W(h_1)U(h_2)\ket{h} & = W(h_1)e^{2\pi i(h,h_2)}\ket{h+h_2}  \\
 & =e^{2\pi i((h_1,h+h_2)+(h,h_2))}\ket{h+h_1+h_2}
\end{split}
\intertext{and}
\begin{split}
U(h_2)W(h_1)\ket{h} & = U(h_2)e^{2\pi i(h_1,h)}\ket{h+h_1}  \\
 & =e^{2\pi i((h_1,h)+(h+h_1,h_2))}\ket{h+h_1+h_2},
\end{split}
\end{align}
therefore $W(h_1)$ and $U(h_2)$ commute for any $h_1,h_2\in H$. Observe that if $h\in H_0$ then $W(h)=U(h)$. It is not difficult to see that both the linear span of $\{W(h)\}_{h\in H}$ and that of $\{U(h)\}_{h\in H}$ are C*-algebras.

Let $\mathfrak{A}'$ be the commutant of $\mathfrak{A}$ in $\boundeds(\mathcal{H})$. Then $\mathfrak{A}\cap\mathfrak{A}'$ contains $\{U(h)\}_{h\in H_0}$. Using these we can form central projections as follows. Let $\chi:H_0\to\mathbb{C}$ be an irreducible character and let
\begin{equation}
P_\chi=\frac{1}{|H_0|}\sum_{h\in H_0}\overline{\chi(h)}W(h).
\end{equation}
Since $W:H\to\boundeds(\mathcal{H})$ restricts to a representation of $H_0$, these are pairwise orthogonal projections, therefore
\begin{equation}
\mathfrak{A}=\bigoplus_{\chi\in\hat{H_0}}P_\chi\mathfrak{A}.
\end{equation}
Here we have $|H_0|=2^c$ terms, therefore it is enough to show that $P_\chi\mathfrak{A}\simeq\boundeds(\mathbb{C}^{2^q})$. Note that $\Tr W(h)=|H|\delta_{h,0}$, so $\Tr P_\chi=|H|/|H_0|$.

Let $\mathcal{H}_\chi$ be the range of $P_\chi$, then the operators $W(h_1)U(h_2)$ restrict to $\mathcal{H}_\chi$. We calculate the Hilbert-Schmidt inner products of these restrictions. Note that for $h\in H_0$ $U(h)=W(h)$ acts as multiplication by $\chi(h)$ on $\mathcal{H}_\chi$, therefore $W(h_1)U(h_2)$ and $W(h'_1)U(h'_2)$ are the same up to a phase factor if $h_1-h'_1\in H_0$ and $h_2-h'_2\in H_0$. Otherwise
\begin{equation}
\begin{split}
& \Tr(P_\chi W(h'_1)U(h'_2))^*(P_\chi W(h_1)U(h_2))  \\
 = & \frac{1}{|H_0|}\sum_{\substack{h\in H  \\ h'\in H_0}}\overline{\chi(h')}e^{2\pi i((h_1,h+h_2)+(h,h_2)-(h'_1,h+h'_2)-(h,h'_2)+(h',h+h_1+h_2))}  \\
 & \cdot\braket{h+h'_1+h'_2}{h+h_1+h_2+h'},
\intertext{which is $0$ unless $h'_1-h_1+h'_2-h_2\in H_0$. If this holds, we let $C=|H_0|^{-1}\chi(h_1-h'_1+h_2-h'_2)e^{2\pi i((h_1,h_2)-(h'_1,h'_2)+(h'_1-h_1+h'_2-h_2,h_1+h_2))}$ and continue as}
& = C\sum_{h\in H}e^{2\pi i((h_1-h'_1,h)+(h,h_2-h'_2)+(h'_1-h_1+h'_2-h_2,h))}  \\
& = C\sum_{h\in H}e^{2\pi i((h,h_2-h'_2)-(h_2-h'_2,h))}.
\end{split}
\end{equation}
This is equal to $0$ iff $h_2-h_2'\notin H_0$. Therefore if $h_1-h'_1\notin H_0$ or $h_2-h'_2\notin H_0$ then $W(h_1)U(h_2)$ is orthogonal to $W(h'_1)U(h'_2)$ (when restricted to $\mathcal{H}_\chi$). It follows that $P_\chi\mathfrak{A}\cup P_\chi\mathfrak{A}'$ spans $\boundeds(\mathcal{H}_\chi)$, therefore the intersection of their commutants (in $\boundeds(\mathcal{H}_\chi)$) consists of multiplies of the idenitity. This means that $P_\chi\mathfrak{A}$ is a factor on a finite dimensional Hilbert space, i.e. it is unitarily equivalent to $\boundeds(\mathbb{C}^n)\otimes I_m$ with $I_m$ the identity on $\mathbb{C}^m$ \cite[Theorem 11.9]{Takesaki}. The appearing dimensions satisfy
\begin{align}
n^2=\dim P_\chi\mathfrak{A}=\frac{|H|}{|H_0|}  \\
m^2=\dim P_\chi\mathfrak{A}'\ge\frac{|H|}{|H_0|}
\intertext{and}
nm=\dim\mathcal{H}_\chi=\Tr P_\chi=\frac{|H|}{|H_0|}
\end{align}
by orthogonality of the generators. But this is only possible if $n=m=\sqrt{\frac{|H|}{|H_0|}}$ holds as claimed.
\end{proof}

This lemma together with the preceding considerations imply the following:
\begin{thm}\label{thm:logicalstates}
Suppose that $H_n(E;\hat{G})$ and $H_{\text{lf}}^n(E;G)$ are finite and let
\begin{align}
\begin{split}
c = & \log_2|\{[a]\in H_n(E;\hat{G})|\forall [b]\in H_\text{lf}^n(E;G):\langle b,a\rangle=0\}|  \\
  & +\log_2|\{[b]\in H_\text{lf}^n(E;G)|\forall [a]\in H_n(E;\hat{G}):\langle b,a\rangle=0\}|
\end{split}
\intertext{and}
q = & \frac{\log_2|H_n(E;\hat{G})\times H_{\text{lf}}^n(E;G)|-c}{2}.
\end{align}
Then
\begin{equation*}
\mathfrak{A}_{\mathrm{logical}}\simeq\mathbb{C}^{2^c}\otimes\boundeds(\mathbb{C}^{2^q}).
\end{equation*}
\end{thm}
\begin{proof}
We only need to observe that (using the notation of Lemma \ref{thm:ccr}) $H_0=\{[a]\in H_n(E;\hat{G})|\forall [b]\in H_\text{lf}^n(E;G):\langle b,a\rangle=0\}\times\{[b]\in H_\text{lf}^n(E;G)|\forall [a]\in H_n(E;\hat{G}):\langle b,a\rangle=0\}$.
\end{proof}
The interpretation of Theorem \ref{thm:logicalstates} is that one can encode $c$ classical bits and $q$ qubits in the ground states.

We would like to stress that in general one should not expect that every ground state is frustration free. Indeed, in section \ref{sec:endomorphisms} we will construct some ground states which are not.

\section{Endomorphisms}\label{sec:endomorphisms}

Our next aim is to look for localized excitations of the system. Recall that in the toric code model on the plane one introduces string operators \cite{Kitaev} for any path (dual path) between pairs of vertices (faces), i.e. the product of $Z$ ($X$) operators acting at the edges on the path. The inner automorphism corresponding to these string operators turns the ground state into a state which can be interpreted as a pair of localized excitations at the endpoints of the path (dual path). A single excitation is obtained by moving one end to infinity \cite{NaaijkensLocalized}.

To make a similar stategy work in our more general model, we need to find analogues of string operators as well as of the process of moving one end to infinity. First observe that a string operator of type $Z$ is actually an operator $Z^b$ for a special $1$-chain $b$, namely the formal sum of edges along a path with coefficients $\pm 1$ depending on the relative orientation of the given $1$-cell and the path. Similarly, a string operator of type $X$ is the same as an operator of the form $X^a$ where $a$ is the sum of edges along a dual path with coefficients $\pm 1$.

Given a path $\gamma$ extending to infinity one can define the sequence of its initial segments $\gamma_k$ and take the limit
\begin{equation}
\omega_\gamma(A):=\lim_{k\to\infty}\omega_0(Z^{\gamma_k}AZ^{-\gamma_k})
\end{equation}
for any local observable $A$. This limit always exists because the sequence on the right hand side is eventually constant. It might seem that it is a crucial property of a path that it is composed of a sequence of edges such that neighbouring ones share an endpoint. However, it is not at all clear what the higher dimensional analogue of this property should be. Fortunately, the only thing which is important is that for any local observable $A$ and large $k,k'$ the equality $\omega_0(Z^{\gamma_k}AZ^{-\gamma_k})=\omega_0(Z^{\gamma_{k'}}AZ^{-\gamma_{k'}})$ holds. This motivates the following definition.
\begin{defn}\label{def:endomorphisms}
Given a locally finite chain $\gamma=\sum_\alpha\chi_\alpha e_\alpha\in C^{\text{lf}}_n(E;\hat{G})$ and a finite subset $K\subseteq\mathcal{E}_n$ we define
\begin{equation}
\gamma_K=\sum_{\alpha\in K}\chi_\alpha e_\alpha.
\end{equation}
This is an element of $C_n(E;\hat{G})$, therefore determines a unitary $Z^{\gamma_K}$. We let $\rho_\gamma:\mathfrak{A}\to\mathfrak{A}$ be the map defined as
\begin{equation}
\rho_\gamma(A):=\lim_{K\to\mathcal{E}_n}Z^{\gamma_K}AZ^{-\gamma_K}.
\end{equation}

Similarly, if $\delta=\sum_\alpha g_\alpha e_\alpha\in C^n(E;G)$ we define
\begin{equation}
\delta_K=\sum_{\alpha\in K}g_\alpha e_\alpha\in C_{\text{lf}}^n(E;G)
\end{equation}
and $\rho_\delta:\mathfrak{A}\to\mathfrak{A}$ as
\begin{equation}
\rho_\delta(A):=\lim_{K\to\mathcal{E}_n}X^{\delta_K}AX^{-\delta_K}.
\end{equation}
\end{defn}

Note that the definiton allows $\gamma$ and $\delta$ to have finite support, therefore it does not really capture the notion of a semi-infinite path. Nevertheless, this construction is well suited for the study of localized excitations and charges. The reason is that if e.g. $|\supp\gamma|<\infty$ then $\rho_\gamma$ turns out to be an inner automorphism, and hence does not change the ``total charge'' of a state.

\begin{prop}\label{prop:endomorphisms}
The limits in definition \ref{def:endomorphisms} exist in the norm topology and $\rho_\gamma$, $\rho_\delta$ are automorphisms. These automorphisms are inner iff $|\supp\gamma|<\infty$ ($|\supp\delta|<\infty$).

The map $(\gamma,\delta)\mapsto\rho_\gamma\circ\rho_\delta$ is a group homomorphism. In particular, $\rho_\gamma$ and $\rho_\delta$ commute.
\end{prop}
\begin{proof}
Let $A$ be an observable and $\varepsilon>0$. Then there is a finite subset $\Lambda_\varepsilon\subseteq\mathcal{E}_n$ and $A_\varepsilon\in\mathfrak{A}_{\Lambda_\varepsilon}$ such that $\norm{A-A_\varepsilon}<\frac{\varepsilon}{2}$. If $\Lambda_\varepsilon\subseteq K\subseteq K'$ then
\begin{equation}
\supp(\gamma_{K'}-\gamma_K)\cap\Lambda_\varepsilon\subseteq (K'\setminus K)\cap\Lambda_\varepsilon=\emptyset,
\end{equation}
therefore $Z^{\gamma_{K'}-\gamma_K}$ and $A_\varepsilon$ commute. Hence
\begin{equation}
\begin{split}
\norm{Z^{\gamma_K}AZ^{-\gamma_K}-Z^{\gamma_{K'}}AZ^{-\gamma_{K'}}}
 & = \norm{A-Z^{\gamma_{K'}-\gamma_K}AZ^{-(\gamma_{K'}-\gamma_K)}}  \\
 & \le\norm{A-A_\varepsilon-Z^{\gamma_{K'}-\gamma_K}(A-A_\varepsilon)Z^{-(\gamma_{K'}-\gamma_K)}}  \\
 & \le 2\norm{A-A_\varepsilon}<\varepsilon,
\end{split}
\end{equation}
estabilishing that the net is Cauchy. $\rho_\gamma$ is an endomorphism since it is a pointwise limit of automorphisms in the norm topology.

The net $K\mapsto Z^{\gamma_K}$ converges iff $|\supp\gamma|<\infty$. By \cite[Theorem 6.3]{EvansKawahigashi} $\rho_\gamma$ is inner iff this happens.

The proof for dual paths is similar.

To verify that we get a homomorphism, it is enough to find the action of the endomorphisms on local observables of the form $X^aZ^b$. For a locally finite chain $\gamma$ we get
\begin{equation}
\begin{split}
\rho_\gamma(X^aZ^b)
 & = \lim_{K\to\mathcal{E}_n}Z^{\gamma_K}X^aZ^bZ^{-\gamma_K}  \\
 & = \lim_{K\to\mathcal{E}_n}e^{2\pi i\langle\gamma_K,a\rangle}X^aZ^bZ^{\gamma_K}Z^{-\gamma_K}  \\
 & = e^{2\pi i\langle\gamma,a\rangle}X^aZ^b
\end{split}
\end{equation}
and similarly, $\rho_\delta(X^aZ^b)=e^{-2\pi i\langle b,\delta\rangle}X^aZ^b$. Therefore $\delta=0$, $\gamma=0$ gives the identity endomorphism and for a pair of locally finite chains $\gamma_1,\gamma_2$ and cochains $\delta_1,\delta_2$ we have
\begin{equation}
\begin{split}
(\rho_{\gamma_1}\circ\rho_{\delta_1})\circ(\rho_{\gamma_2}\circ\rho_{\delta_2})(X^aZ^b)
 & = (\rho_{\gamma_1}\circ\rho_{\delta_1})(e^{2\pi i\langle\gamma_2,a\rangle-2\pi i\langle b, \delta_2\rangle}X^aZ^b)  \\
 & = e^{2\pi i\langle\gamma_1+\gamma_2,a\rangle-2\pi i\langle b,\delta_1+\delta_2\rangle}X^aZ^b  \\
 & = (\rho_{\gamma_1+\gamma_2}\circ\rho_{\delta_1+\delta_2})(X^aZ^b)
\end{split}
\end{equation}
as required.
\end{proof}

We can write the composition of the two types of endomorphisms in the following way, which is sometimes useful:
\begin{equation}
(\rho_\gamma\circ\rho_\delta)(A)=\lim_{K\to\mathcal{E}_n}Z^{\gamma_K}X^{\delta_K}AX^{-\delta_K}Z^{-\gamma_K}.
\end{equation}
We will also use the notation $\rho_{(\gamma,\delta)}=\rho_\gamma\circ\rho_\delta$.

Endomorphisms can be used to create new states from old ones by composition. Since our endomorphisms are in fact invertible, a simple observation is that this construction preserves purity, i.e. $\omega$ is pure iff $\omega\circ\rho_{(\gamma,\delta)}$ is pure.

Suppose that we are given a fixed frustration free ground state $\omega_0$. For any $\gamma\in C^{\text{lf}}_n(E;\hat{G})$ and $\delta\in C^n(E;G)$ we can define a state $\omega_{(\gamma,\delta)}=\omega_0\circ\rho_{(\gamma,\delta)}$.\footnote{Note that this state also depends on the chosen state $\omega_0$ even though the notation does not reflect this.} The expected value of an observable of the form $X^aZ^b$ can be computed as
\begin{equation}\label{eq:transformedstate}
\omega_{(\gamma,\delta)}(X^aZ^b) = e^{2\pi i(\langle\gamma,a\rangle-\langle b,\delta\rangle)}\omega_0(X^aZ^b).
\end{equation}
In particular, if an observable $A$ is localized in $\mathcal{E}_n\setminus(\supp\gamma\cup\supp\delta)$ then $\omega_{(\gamma,\delta)}(A)=\omega_0(A)$.

For the toric code it is known that the resulting state only depends on the starting point and not on the path along which the other endpoint is moved away \cite[Lemma 3.6]{NaaijkensLocalized}. Therefore we expect that also in our model different chains and cochains may give rise to the same state. The following proposition confirms this expectation:
\begin{prop}\label{prop:sameexcitation}
For a pair of locally finite chains $\gamma,\gamma'$ and a pair of cochains $\delta,\delta'$ the following conditions are equivalent:
\begin{enumerate}
\item\label{itm:difforthogonal} for any $n$-cycle $b$ we have $\langle b,\delta'-\delta\rangle=0$ and for any locally finite $n$-cocycle $a$ we have $\langle \gamma'-\gamma,a\rangle=0$.
\item\label{itm:statesequal} for any frustration free ground state $\omega_0:\mathfrak{A}\to\mathbb{C}$ the equality $\omega_0\circ\rho_{(\gamma,\delta)}=\omega_0\circ\rho_{(\gamma',\delta')}$ holds.
\end{enumerate}
\end{prop}
\begin{proof}
Suppose that condition \ref{itm:difforthogonal} holds and let $\omega_0$ be any frustration free ground state. Then for any $a\in C^n_{\text{lf}}(E;G)$ and $b\in C_n(E;\hat{G})$ the equality
\begin{equation}
\begin{split}
\omega_{(\gamma',\delta')}(X^aZ^b)
 & = e^{2\pi i(\langle\gamma',a\rangle-\langle b,\delta'\rangle)}\omega_0(X^aZ^b)  \\
 & = e^{2\pi i(\langle\gamma',a\rangle-\langle b,\delta'\rangle)}e^{-2\pi i(\langle\gamma,a\rangle-\langle b,\delta\rangle)}\omega_{(\gamma,\delta)}(X^aZ^b)  \\
 & = e^{2\pi i(\langle\gamma'-\gamma,a\rangle-\langle b,\delta'-\delta\rangle)}\omega_{(\gamma,\delta)}(X^aZ^b)
\end{split}
\end{equation}
holds. If $\partial^T a=0$ and $\partial b=0$ then by assumption the exponent is $0$, therefore $\omega_{(\gamma',\delta')}(X^aZ^b)=\omega_{(\gamma,\delta)}(X^aZ^b)$. Otherwise $\omega_0(X^aZ^b)=0$ and hence also $\omega_{(\gamma',\delta')}(X^aZ^b)=\omega_{(\gamma,\delta)}(X^aZ^b)$. We can conclude that the two states agree on local observables, and by continuity, they coincide.

Conversely, for any $a\in C^n_{\text{lf}}(E;G)$ and $b\in C_n(E;\hat{G})$ satisfying $\partial^T a=0$ and $\partial b=0$ there exists a frustration free ground state $\omega_0$ such that $\omega_0(X^aZ^b)\neq 0$. By the equality above one must have $\langle\gamma'-\gamma,a\rangle-\langle b,\delta'-\delta\rangle=0$, but this is only possible if condition \ref{itm:difforthogonal} is true.
\end{proof}

It is instructive to see how this characterization boils down in the toric code case. Take two paths $\gamma,\gamma'$ with common starting point and extending to infinity (see Figure \ref{fig:twopaths}). Since $\partial(\gamma'-\gamma)=0$ and the first locally finite homology of the plane is $0$ (with arbitrary coefficients), there is a locally finite $2$-chain $p$ such that $\gamma'-\gamma=\partial p$. But then for any locally finite $1$-cocycle $a$ we have $\langle \gamma'-\gamma,a\rangle=\langle\partial p,a\rangle=\langle p,\partial^Ta\rangle=0$. On the other hand if the two starting points differ then $\partial(\gamma'-\gamma)\neq 0$ and therefore one can find a locally finite $n-1$-cochain $c$ such that $\langle\partial(\gamma'-\gamma),c\rangle\neq 0$. In this case choose $a=\partial^Tc$ so that $\langle\gamma'-\gamma,a\rangle=\langle\partial(\gamma'-\gamma),c\rangle\neq 0$.
\begin{figure}
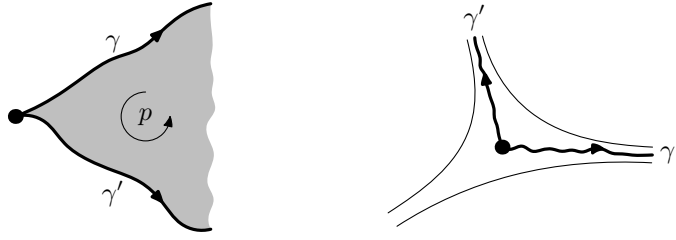

\begin{subfigure}{.45\textwidth}
\centering
\includegraphics{twopaths1.mps}
\end{subfigure}
\begin{subfigure}{.45\textwidth}
\centering
\includegraphics{twopaths2.mps}
\end{subfigure}
\caption{Two semi-infinite paths starting at the same point. On the plane (left) we can always find a locally finite $2$-chain $p$ such that $\partial p=\gamma'-\gamma$. If the end space of the CW complex is disconnected (right), then this may fail to be true, even when the CW complex is contractible. ($n=1$)\label{fig:twopaths}}
\end{figure}

For the toric code the starting point of an infinite path has a distinguished role, since observables supported near this point can detect the presence of the corresponding excitation, but interaction terms corresponding to distant vertices or faces cannot. In our model a similar role is played by the supports of $\partial\gamma$ and $\partial^T\delta$. More precisely, from the proof of Proposition \ref{prop:sameexcitation} one can see that when $a=\partial^Tc$ and $b=\partial d$ where $\supp c\cap\supp\partial(\gamma'-\gamma)=\supp d\cap\supp\partial^T(\delta'-\delta)=\emptyset$ then $(\omega_0\circ\rho_{(\gamma,\delta)})(X^aZ^b)=(\omega_0\circ\rho_{(\gamma',\delta')})(X^aZ^b)$.

In particular, if $\partial\gamma=0$ and $\partial^T\delta=0$ then no observable localized in any compact contractible subcomplex changes its expected value compared to the starting state corresponding to $\gamma'=\delta'=0$ (see Figure \ref{fig:hole}). This condition is satisfied by the interaction terms, so in this case the resulting state is still frustration free.
\begin{prop}\label{prop:tostrongground}
Let $\omega_0$ be a frustration free ground state, $\gamma$ a locally finite $n$-chain and $\delta$ an $n$-cochain. Then $\omega_{(\gamma,\delta)}$ is a frustration free iff $\partial\gamma=0$ and $\partial^T\delta=0$.
\end{prop}
\begin{proof}
Recall that $\omega_{(\gamma,\delta)}$ is a frustration free ground state iff $\omega_{(\gamma,\delta)}(X^{\partial^Tc}Z^{\partial d})=1$ for every $c\in C_{\text{lf}}^{n-1}(E;G)$ and $d\in C_{n+1}(E;\hat{G})$. But since
\begin{equation}
\omega_{(\gamma,\delta)}(X^{\partial^Tc}Z^{\partial d})=e^{2\pi i(\langle\gamma,\partial^Tc\rangle-\langle\partial d,\delta\rangle)}\omega_0(X^{\partial^Tc}Z^{\partial d})=e^{2\pi i(\langle\partial\gamma,c\rangle-\langle d,\partial^T\delta\rangle)},
\end{equation}
this holds iff $\partial\gamma=0$ and $\partial^T\delta=0$.
\end{proof}
Proposition \ref{prop:sameexcitation} implies that if $\gamma$ is a locally finite $n$-cycle and $\delta$ is an $n$-cocycle then $\omega_{(\gamma,\delta)}$ only depends on $[\gamma]\in H_n^{\text{lf}}(E;\hat{G})$ and $[\delta]\in H^n(E;G)$.
\begin{figure}
\centering
\includegraphics{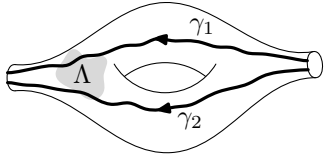}
\caption{On this open surface the two strings $\gamma_1$ and $\gamma_2$ are both locally finite cycles, but represent different locally finite homology classes. An observable $A$ localized in the compact contractible region $\Lambda$ cannot detect the difference: $\omega_0(A)=\omega_{(\gamma_1,0)}(A)=\omega_{(\gamma_2,0)}(A)$. ($n=1$)\label{fig:hole}}
\end{figure}

Initially, we introduced the endomorphisms $\rho_\gamma$ and $\rho_\delta$ in the hope of finding states describing elementary excitations. Interestingly, they are also useful for finding ground states that are not frustration free.
\begin{prop}\label{prop:toground}
Let $\omega_0$ be a frustration free ground state, $\gamma\in C_n^{\text{lf}}(E;\hat{G})$ and $\delta\in C^n(E;G)$. Then $\omega_{(\gamma,\delta)}$ is a ground state iff for all $b\in C_n(E;\hat{G})$
\begin{equation}
\lim_{K_-\to\mathcal{E}_{n-1}} |\supp(\partial(\gamma+b))_{K_-}|-|(\supp\partial\gamma)_{K_-}|>0
\end{equation}
and for all $a\in C^n_{\text{lf}}(E;G)$
\begin{equation}
\lim_{K_+\to\mathcal{E}_{n+1}} |\supp(\partial^T(\delta+a))_{K_+}|-|(\supp\partial^T\delta)_{K_+}|>0
\end{equation}
holds.
\end{prop}
\begin{proof}
We need to compute $-i\omega_{(\gamma,\delta)}(A^*\delta(A))$ for an arbitrary local $A$. Let $A=\sum_ic_iX^{a_i}Z^{b_i}$ where the sum is over some finite index set. Then for $\alpha\in\mathcal{E}_{n-1}$ we have
\begin{equation}
\begin{split}
&-\frac{1}{|G|}\sum_{g\in G}\omega_{(\gamma,\delta)}(Z^{-b_j}X^{-a_j}[X^{\partial^T(g e_\alpha)},X^{a_i}Z^{b_i}])  \\
 & = -\frac{1}{|G|}\sum_{g\in G}e^{2\pi i(\langle\gamma,a_i-a_j+\partial^T(g e_\alpha)\rangle-\langle b_i-b_j,\delta\rangle)}\omega_0(Z^{-b_j}X^{-a_j}[X^{\partial^T(g e_\alpha)},X^{a_i}Z^{b_i}])  \\
 & = -\frac{1}{|G|}\sum_{g\in G}e^{2\pi i(\langle\gamma,a_i-a_j+\partial^T(g e_\alpha)\rangle-\langle b_i-b_j,\delta\rangle)}\omega_0(Z^{-b_j}X^{-a_j}X^{a_i}Z^{b_i})(e^{-2\pi i\langle b_i,\partial^T(g e_\alpha)\rangle}-1)  \\
 & = -\frac{1}{|G|}\sum_{g\in G}e^{2\pi i(\langle\gamma,a_i-a_j\rangle-\langle b_i-b_j,\delta\rangle)}\omega_0(Z^{-b_j}X^{-a_j}X^{a_i}Z^{b_i})(e^{-2\pi i\langle\partial(\gamma-b_i),g e_\alpha\rangle}-e^{2\pi i\langle\partial\gamma,g e_\alpha})  \\
 & = -e^{2\pi i(\langle\gamma,a_i-a_j\rangle-\langle b_i-b_j,\delta\rangle)}\omega_0(Z^{-b_j}X^{-a_j}X^{a_i}Z^{b_i})\left((1-1_{\supp\partial(\gamma-b_i)}(\alpha))-(1-1_{\supp\partial\gamma}(\alpha))\right)  \\
 & = e^{2\pi i(\langle\gamma,a_i-a_j\rangle-\langle b_i-b_j,\delta\rangle)}\omega_0(Z^{-b_j}X^{-a_j}X^{a_i}Z^{b_i})\left(1_{\supp\partial(\gamma-b_i)}(\alpha)-1_{\supp\partial\gamma}(\alpha)\right),
\end{split}
\end{equation}
where $1_A$ is the indicator function of the set $A$, and similarly for $\beta\in\mathcal{E}_{n+1}$ the equality
\begin{multline}
-\frac{1}{|G|}\sum_{\chi\in \hat{G}}\omega_{(\gamma,\delta)}(Z^{-b_j}X^{-a_j}[Z^{\partial(\chi e_\beta)},X^{a_i}Z^{b_i}])  \\
= e^{2\pi i(\langle\gamma,a_i-a_j\rangle-\langle b_i-b_j,\delta\rangle)}\omega_0(Z^{-b_j}X^{-a_j}X^{a_i}Z^{b_i})\left(1_{\supp\partial^T(\delta-a_i)}(\beta)-1_{\supp\partial^T\delta}(\beta)\right)
\end{multline}
holds. Let us introduce the notation $c'_i=e^{2\pi i(\langle\gamma,a_i\rangle-\langle b_i,\delta\rangle)}c_i$. Using these we compute
\begin{equation}
\begin{split}
& -i\omega_{(\gamma,\delta)}(A^*\delta(A))  \\
 = & -i\omega_{(\gamma,\delta)}(A^*\lim_{\Lambda\to\mathcal{E}_n}i[H_\Lambda,A])  \\
 = & -\sum_{\substack{\alpha\in\mathcal{E}_{n-1}  \\  \beta\in\mathcal{E}_{n+1}}}\frac{1}{|G|}\sum_{\substack{g\in G  \\  \chi\in\hat{G}}}\omega_{(\gamma,\delta)}(A^*[X^{\partial^T(g e_\alpha)},A])+\omega_{(\gamma,\delta)}(A^*[Z^{\partial(\chi e_\beta)},A])  \\
 = & \sum_{i,j}\overline{c'_j}c'_i\omega_0(Z^{-b_j}X^{-a_j}X^{a_i}Z^{b_i})\Big(\lim_{K_-\to\mathcal{E}_{n-1}}
|\supp(\partial(\gamma-b_i))_{K_-}|-|\supp(\partial\gamma)_{K_-}|\\
 & +\lim_{K_+\to\mathcal{E}_{n+1}}|\supp(\partial^T(\delta-a_i))_{K_+}|-|\supp(\partial^T\delta)_{K_+}|\Big),
\end{split}
\end{equation}
which is a quadratic form associated to a matrix. This matrix is actually block-diagonal in the sense that its $i,j$ entry is $0$ when $\partial^T(a_i-a_j)\neq 0$ or $\partial(b_i-b_j)\neq 0$. Each block is a nonzero positive semidefinite matrix with entries $\omega_0(Z^{-b_j}X^{-a_j}X^{a_i}Z^{b_i})$ multiplied by the constant (depending only on the block but otherwise not on $i$ and $j$)
\begin{multline}
\lim_{K_-\to\mathcal{E}_{n-1}}
|\supp(\partial(\gamma-b_i))_{K_-}|-|\supp(\partial\gamma)_{K_-}|  \\  +\lim_{K_+\to\mathcal{E}_{n+1}}|\supp(\partial^T(\delta-a_i))_{K_+}|-|\supp(\partial^T\delta)_{K_+}|,
\end{multline}
therefore positivity of this expression for any $a_i$ and $b_i$ is a necessary and sufficient condition for $\omega_{(\gamma,\delta)}$ to be a ground state. Taking separately $a_i=0$ or $b_i=0$ gives the statement.
\end{proof}

This result implies that in general, our model has many pure ground states. In fact, already Kitaev's toric code model on the plane has infinitely many pure ground states. Namely, if we take $\gamma$ and $\delta$ to be such that $|\supp\partial\gamma|=|\supp\partial^T\delta|=1$ then it is not possible to find $c$ and $d$ with finite support such that $\supp\partial(\gamma+c)=\supp\partial^T(\delta+d)=\emptyset$, since $|\supp\partial c|\ge 2$ and $|\supp\partial^T d|\ge 2$ whenever $c\neq 0$ and $d\neq 0$. Once such $\gamma$ and $\delta$ are chosen, one can translate them to any part of the plane giving rise to distinct ground states.

More generally, if $|\supp\partial\gamma|<\infty$ and $|\supp\partial^T\delta|<\infty$ then the limits in Proposition \ref{prop:toground} reduce to $|\supp\partial(\gamma+b)|-|\supp\partial\gamma|$ and $|\supp\partial^T(\delta+a)|-|\supp\partial^T\delta|$. Since the supports are natural numbers, it is possible to find a $b$ ($a$) such that $|\supp\partial(\gamma+b)|$ ($|\supp\partial^T(\delta+a)|$) is minimal. Then $\omega_{(\gamma+b,\delta+a)}$ is a ground state and unitarily equivalent to $\omega_{(\gamma,\delta)}$ (see Proposition \ref{thm:equivalent}).

\section{GNS representations}\label{sec:GNSreps}

In this section we study GNS representations corresponding to the states encountered in the preceding sections. Recall that given a state $\omega$ on a C*-algebra $\mathfrak{A}$ the GNS construction yields a representation $\pi_\omega:\mathfrak{A}\to\boundeds(\mathcal{H}_\omega)$ together with a cyclic vector $\Omega_\omega\in\mathcal{H}_\omega$ such that $\omega(A)=\langle\Omega_\omega,\pi_\omega(A)\Omega_\omega\rangle$. Such a triple $(\mathcal{H}_\omega,\pi_\omega,\Omega_\omega)$ is unique up to unitary equivalence.

Two states $\omega_1$ and $\omega_2$ are said to be quasi-equivalent if their GNS representations $\pi_{\omega_1}$ and $\pi_{\omega_2}$ are quasi-equivalent.

Once we have a representation of $\mathfrak{A}$ we can take the weak closure (equivalently: double commutant) of the set of representing operators to get a von~Neumann algebra. The state in question is called a factor state if this von~Neumann algebra is a factor.

Our first aim is to study GNS representations of frustration free ground states. Recall that such a state $\omega:\mathfrak{A}\to\mathbb{C}$ is uniquely determined by a state $\tilde{\omega}:\mathfrak{A}_{\mathrm{logical}}\to\mathbb{C}$ satisfying $\omega|_{\mathfrak{A}_{XZ}'}=\tilde{\omega}\circ p$ with $p:\mathfrak{A}_{XZ}'\to\mathfrak{A}_{XZ}'/J=\mathfrak{A}_{\mathrm{logical}}$ the canonical projection. The following theorem clarifies the relation between the GNS representations of $\omega$ and $\tilde{\omega}$.
\begin{thm}\label{thm:stronggroundGNS}\leavevmode
\begin{enumerate}
\item $\pi_\omega$ is irreducible iff $\pi_{\tilde{\omega}}$ is irreducible.
\item $\pi_\omega(\mathfrak{A})''\cap\pi_\omega(\mathfrak{A})'\simeq\pi_{\tilde{\omega}}(\mathfrak{A}_{XZ}'/J)''\cap\pi_{\tilde{\omega}}(\mathfrak{A}_{XZ}'/J)'$. In particular, $\omega$ is a factor state iff $\tilde{\omega}$ is a factor state.
\item two factor states $\omega_1$ and $\omega_2$ are quasi-equivalent iff the corresponding states $\tilde{\omega}_1$ and $\tilde{\omega}_2$ are quasi-equivalent.
\end{enumerate}
\end{thm}
\begin{proof}
We have seen that $\omega$ is pure iff $\tilde{\omega}$ is pure. The first statement follows because a state is pure iff its GNS representation is irreducible.

The third statement can be proved using the second one as follows. Recall that the map $\omega\mapsto\tilde{\omega}$ preserves convex combinations. The factor states $\omega_1$ and $\omega_2$ are quasi-equivalent iff $\frac{1}{2}(\omega_1+\omega_2)$ is a factor state as well. This, in turn, is equivalent to $\frac{1}{2}(\tilde{\omega}_1+\tilde{\omega}_2)$ being a factor state, which is true iff $\tilde{\omega_1}$ and $\tilde{\omega_2}$ are quasi-equivalent.

Now we prove the second statement. Let $\omega$ be a frustration free ground state and $\omega|_{\mathfrak{A}_{XZ}'}=\tilde{\omega}\circ p$. First we show that $\mathcal{H}_{\tilde{\omega}}$ can be viewed as a subspace of $\mathcal{H}_{\omega}$. To see this, recall that in the GNS construction one introduces a semidefinite inner product on $\mathfrak{A}$ as $\langle A,B\rangle_\omega=\omega(A^*B)$. The set $J_{\omega}=\{A\in\mathfrak{A}|\langle A,A\rangle_\omega=0\}$ is the Gelfand ideal corresponding to $\omega$ and $\mathcal{H}_\omega$ is the completion of $\mathfrak{A}/J_\omega$ with respect to the inner product induced by $\langle\cdot,\cdot\rangle_\omega$. The cyclic vector is $[I]$ and the representation is given by (the continuous extension of) $\pi_\omega(A)[B]=[AB]$.

Similarly, for $A,B\in\mathfrak{A}_{XZ}'$ we have
\begin{equation}
\langle p(A),p(B)\rangle_{\tilde{\omega}}=\tilde{\omega}(p(A)^*p(B))=\tilde{\omega}(p(A^*B))=\omega(A^*B)
\end{equation}
and $J_{\tilde{\omega}}=\{p(A)|A\in\mathfrak{A}_{XZ}', \omega(A^*A)=0\}$. $\mathcal{H}_{\tilde{\omega}}$ is the completion of $(\mathfrak{A}_{XZ}'/J)/J_{\tilde{\omega}}$. Note that $J_{\tilde{\omega}}=p(J_\omega\cap\mathfrak{A}_{XZ}')$.

Observe that $J\le J_\omega$, as can be seen from the fact that $J$ is a *-algebra and $\omega|_J=0$. This implies that
\begin{equation}
(\mathfrak{A}_{XZ}'/J)/J_{\tilde{\omega}}\simeq(\mathfrak{A}_{XZ}'/J)/(J_\omega\cap\mathfrak{A}_{XZ}'/J)\simeq\mathfrak{A}_{XZ}'/(J_\omega\cap\mathfrak{A}_{XZ}')\le\mathfrak{A}/J_\omega
\end{equation}
as claimed. With this identification we also have $\pi_{\tilde{\omega}}(p(A))[B]=[AB]=\pi_\omega(A)[B]$ for any $A,B\in\mathfrak{A}_{XZ}'$.

Let $P$ be the orthogonal projection to $\mathcal{H}_{\tilde{\omega}}$. Then $P$ can be written as the product of all the star and plaquette operators. More precisely, we now show that
\begin{equation}
P=\lim_{\Lambda_\pm\to\mathcal{E}_{n\pm 1}}\prod_{\alpha\in\Lambda_-}\pi_\omega(A_\alpha)\prod_{\beta\in\Lambda_+}\pi_\omega(B_\beta),
\end{equation}
where the limit exists in the strong operator topology. Since the $A_\alpha$ and $B_\beta$ are projections, the norm of products of them does not exceed $1$. It is therefore enough to verify convergence on a dense set. Let $a\in C_{\text{lf}}^n(E;G)$ and $b\in C_n(E;\hat{G})$. If $\alpha\in\mathcal{E}_{n-1}$ then
\begin{equation}
\begin{split}
\pi_\omega(A_\alpha)[X^aZ^b]
 & = \frac{1}{|G|}\sum_{g\in G}[X^{\partial^T(ge_\alpha)}X^aZ^b]  \\
 & = \frac{1}{|G|}\sum_{g\in G}e^{-2\pi i\langle b,\partial^T(ge_\alpha)\rangle}[X^aZ^bX^{\partial^T(ge_\alpha)}]  \\
 & = \frac{1}{|G|}\sum_{g\in G}e^{-2\pi i\langle\partial b,ge_\alpha\rangle}[X^aZ^b]  \\
 & = \begin{cases}
[X^aZ^b] & \text{if }(\partial b)_\alpha=0  \\
0 & \text{otherwise}
\end{cases}
\end{split}
\end{equation}
and similarly for $\beta\in\mathcal{E}_{n+1}$. These imply that if $\supp\partial^Ta\subseteq\Lambda_+$ and $\supp\partial b\subseteq\Lambda_-$ then
\begin{equation}
\prod_{\alpha\in\Lambda_-}\pi_\omega(A_\alpha)\prod_{\beta\in\Lambda_+}\pi_\omega(B_\beta)[X^aZ^b]=\begin{cases}
[X^aZ^b] & \text{if }\partial b=\partial^Ta=0  \\
0 & \text{otherwise}
\end{cases},
\end{equation}
which is the same as $P[X^aZ^b]$. We can conclude that $P\in\pi_{\omega}(\mathfrak{A})''\cap\pi_{\omega}(\mathfrak{A}_{XZ}')'$.

For any $A\in\mathfrak{A}$ the operator $P\pi_\omega(A)P$ maps $\mathcal{H}_{\tilde{\omega}}$ into itself. In fact, such operators are in the range of $\pi_{\tilde{\omega}}$. To see this, it is enough to compute the action of such an operator with $A=X^aZ^b$ on a vector $[X^{a'}Z^{b'}]$ with $\partial^Ta'=\partial b'=0$:
\begin{equation}
\begin{split}
P\pi_\omega(X^aZ^b)P[X^{a'}Z^{b'}]
 & = P[X^aZ^bX^{a'}Z^{b'}]  \\
 & = e^{2\pi i\langle b,a'\rangle}P[X^{a+a'}Z^{b+b'}]  \\
 & = \begin{cases}
\pi_{\tilde{\omega}}(p(X^aZ^b))[X^{a'}Z^{b'}] & \text{if }\partial^Ta=\partial b=0  \\
0 & \text{otherwise}
\end{cases}.
\end{split}
\end{equation}

By taking limits we can conclude that if $C\in\pi_{\omega}(\mathfrak{A})''\cap\pi_{\omega}(\mathfrak{A})'$, then $\tilde{C}:=PCP\in\pi_{\tilde{\omega}}(\mathfrak{A}_{XZ}'/J)''\cap\pi_{\tilde{\omega}}(\mathfrak{A}_{XZ}'/J)'$. Moreover, the map $C\mapsto\tilde{C}$ is a homomorphism since $(PC_1P)(PC_2P)=P^2C_1C_2P^2=PC_1C_2P$.

In the other direction, suppose that $\tilde{C}\in\pi_{\tilde{\omega}}(\mathfrak{A}_{XZ}'/J)''\cap\pi_{\tilde{\omega}}(\mathfrak{A}_{XZ}'/J)'$. Then we claim that $C[A]:=\pi_\omega(A)\tilde{C}[I]$ gives a well-defined operator on $\mathcal{H}_\omega$. To see this, suppose that $A\in J_\omega$. For any $\varepsilon>0$ there exists a local observable $A_\varepsilon$ such that $\norm{A-A_\varepsilon}<\varepsilon$. Then
\begin{equation}
\begin{split}
\omega(A_\varepsilon^*A_\varepsilon)
 & = \omega((A-(A-A_\varepsilon))^*(A-(A-A_\varepsilon)))  \\
 & = \omega(A^*A)+\omega((A-A_\varepsilon)^*(A-A_\varepsilon))+2\Re\omega(A^*(A-A_\varepsilon))  \\
 & \le \varepsilon^2+2\norm{A}\varepsilon.
\end{split}
\end{equation}
On the other hand, if we write $A_\varepsilon=\sum_{i\in I}\alpha_iX^{a_i}Z^{b_i}$ for some finite index set $I$ then
\begin{equation}
\omega(A_\varepsilon^*A_\varepsilon)
= \sum_{i,j\in I}\alpha_i\overline{\alpha_j} e^{2\pi i\langle -b_j,a_i-a_j\rangle}\omega(X^{a_i-a_j}Z^{b_i-b_j}).
\end{equation}
Here the last factor is $0$ unless $\partial^T(a_i-a_j)=\partial(b_i-b_j)=0$. Let us group the terms according to the values of $\partial^Ta_i$ and $\partial b_i$ and write $A_\varepsilon=\sum_{k\in K}A_{\varepsilon,k}$ with $A_{\varepsilon,k}=\sum_{i\in I_k}\alpha_iX^{a_i}Z^{b_i}$. Then
\begin{equation}
\omega(A_\varepsilon^*A_\varepsilon)=\sum_{k\in K}\omega(A_{\varepsilon,k}^*A_{\varepsilon,k}).
\end{equation}
For each $k\in K$ let us fix some index $i_k\in I_k$. Then $Z^{-b_{i_k}}X^{-a_{i_k}}A_{\varepsilon,k}\in\mathfrak{A}_{XZ}'$ commutes with $\tilde{C}$, therefore we can write
\begin{equation}
\begin{split}
\pi_\omega(A_\varepsilon)\tilde{C}[I]
 & = \sum_{k\in K}\pi_\omega(X^{a_{i_k}}Z^{b_{i_k}})\pi_\omega(Z^{-b_{i_k}}X^{-a_{i_k}}A_{\varepsilon,k})\tilde{C}[I]  \\
 & = \sum_{k\in K}\pi_\omega(X^{a_{i_k}}Z^{b_{i_k}})\tilde{C}\pi_{\tilde{\omega}}(p(Z^{-b_{i_k}}X^{-a_{i_k}}A_{\varepsilon,k}))[I].
\end{split}
\end{equation}

Let $(C_n)_{n\in N}$ be a net in $\mathfrak{A}_{XZ}'$ such that $\pi_{\tilde{\omega}}(p(C_n))\to\tilde{C}$ in the strong operator topology. Using continuity of multiplication we have
\begin{equation}
\begin{split}
\norm{\pi_\omega(A_\varepsilon)\tilde{C}[I]}^2
 & = \lim_n\norm{\sum_{k\in K}\pi_\omega(X^{a_{i_k}}Z^{b_{i_k}})\pi_\omega(C_n)\pi_{\tilde{\omega}}(p(Z^{-b_{i_k}}X^{-a_{i_k}}A_{\varepsilon,k}))[I]}^2  \\
 & = \lim_n\sum_{k,k'\in K}\omega\left(A_{\varepsilon,k'}X^{a_{i_{k'}}}Z^{b_{i_{k'}}}C_n^*Z^{-b_{i_{k'}}}X^{-a_{i_{k'}}}X^{a_{i_{k}}}Z^{b_{i_{k}}}C_nZ^{-b_{i_{k}}}X^{-a_{i_{k}}}A_{\varepsilon,k}\right)  \\
 & = \lim_n\sum_{k\in K}\omega\left(A_{\varepsilon,k}X^{a_{i_k}}Z^{b_{i_k}}C_n^*C_nZ^{-b_{i_k}}X^{-a_{i_k}}A_{\varepsilon,k}\right)  \\
 & = \sum_{k\in K}\lim_n\norm{\pi_\omega(C_n)Z^{-b_{i_k}}X^{-a_{i_k}}A_{\varepsilon,k}[I]}^2  \\
 & = \sum_{k\in K}\norm{\tilde{C}Z^{-b_{i_k}}X^{-a_{i_k}}A_{\varepsilon,k}[I]}^2  \\
 & \le \sum_{k\in K}\norm{\tilde{C}}^2\norm{Z^{-b_{i_k}}X^{-a_{i_k}}A_{\varepsilon,k}[I]}^2  \\
 & = \norm{\tilde{C}}^2\sum_{k\in K}\omega(A_{\varepsilon,k}^*A_{\varepsilon,k})\le \norm{\tilde{C}}^2(\varepsilon^2+2\norm{A}\varepsilon).
\end{split}
\end{equation}
This can be made arbitrarily small as $\varepsilon\to 0$, therefore $\pi_\omega(A)\tilde{C}[I]=0$.

By construction, $C\in\pi_\omega(\mathfrak{A})'$. We show that it is in the strong closure as well by exhibiting a net converging to $C$. To this end consider the subspaces $\overline{\pi_\omega(X^aZ^b\mathfrak{A}_{XZ}')}\le\mathcal{H}_\omega$. It is not difficult to see that two such subspaces are either identical or orthogonal depending on the boundary of $b$ and the coboundary of $a$. Moreover, with $\tilde{C}$ as before the operator $\pi_\omega(X^aZ^b)\tilde{C}P\pi_\omega(Z^{-b}X^{-a})$ is supported on $\overline{\pi_\omega(X^aZ^b\mathfrak{A}_{XZ}')}\le\mathcal{H}_\omega$ and remains unchanged if we add a locally finite cocycle $a'$ to $a$ or a cycle $b'$ to $b$:
\begin{multline}
\pi_\omega(X^{a+a'}Z^{b+b'})\tilde{C}P\pi_\omega(Z^{-b-b'}X^{-a-a'})
\\ = e^{-2\pi i\langle b,a'\rangle}X^aZ^bX^{a'}Z^{b'}\tilde{C}Pe^{2\pi i\langle -b,-a'\rangle}Z^{-b'}X^{-a'}Z^{-b}X^{-a}\\
\\ = \pi_\omega(X^aZ^b)\tilde{C}P\pi_\omega(Z^{-b}X^{-a})
\end{multline}
since $\tilde{C}P$ commutes with $X^{a'}Z^{b'}$. Let us denote the corresponding operator by $C(a+Z_n,b+Z_\text{lf}^n)$, where $Z_n$ and $Z_\text{lf}^n$ stands for the group of $n$-cycles and locally finite $n$-cocycles, respectively. Since $\tilde{C}\in\pi_\omega(\mathfrak{A})$, the same holds for $C(a+Z_n,b+Z_\text{lf}^n)$. We will show that
\begin{equation}
\sum_{\substack{a+Z_n\in C_n/Z_n  \\  b+Z_{\text{lf}}^n\in C_{\text{lf}}^n/Z_{\text{lf}}^n}}C(a+Z_n,b+Z_\text{lf}^n)=C,
\end{equation}
where the limit exists in the strong operator topology. First note that different terms have orthogonal supports and each term has norm at most $\norm{C}$, therefore the sum over any finite subset also has norm at most $\norm{C}$. This means that it is enough to verify convergence on a set of vectors which span a dense subspace. For $[X^{a'}Z^{b'}]$ we have
\begin{multline}
C(a+Z_n,b+Z_\text{lf}^n)[X^{a'}Z^{b'}]
  = \pi_\omega(X^aZ^b)\tilde{C}P\pi_\omega(Z^{-b}Z^{-a})[X^{a'}Z^{b'}]  \\
  = \begin{cases}
0 & \text{if }\partial^Ta'\neq \partial^Ta \text{ or } \partial b'\neq \partial b  \\
\pi_\omega(X^{a'}Z^{b'})\tilde{C}[I] & \text{otherwise}
\end{cases}
\end{multline}
after replacing $a$ and $b$ by $a'$ and $b'$ when they belong to the same class. Therefore the sum over the classes contains only one term and it is equal to $C[X^{a'}Z^{b'}]$.

Finally, it is clear from the constructions that the maps $C\mapsto \tilde{C}$ and $\tilde{C}\mapsto C$ are inverses of each other, therefore $\pi_\omega(\mathfrak{A})''\cap\pi_\omega(\mathfrak{A})'\simeq\pi_{\tilde{\omega}}(\mathfrak{A}_{XZ}'/J)''\cap\pi_{\tilde{\omega}}(\mathfrak{A}_{XZ}'/J)'$.
\end{proof}

Next we fix a pure frustration free ground state $\omega_0$ and ask the following question: When do two states $\omega_{(\gamma,\delta)}=\omega_0\circ\rho_{(\gamma,\delta)}$ and $\omega_{(\gamma',\delta')}$ give rise to equivalent GNS representations? Here we will assume that $\partial\gamma\in C_{n-1}(E;\hat{G})$ and $\partial^T\delta\in C_{\text{lf}}^{n+1}(E;G)$. Recall that in this case $(-\partial\gamma,\gamma)$ and $(-\partial^T\delta,\delta)$ represent elements of $H^\infty_{n-1}(E;\hat{G})$ and $H_\infty^n(E;G)$, respectively. The following proposition shows that the equivalence class of the GNS representation depends only on these homology and cohomology classes:
\begin{prop}\label{thm:equivalent}
Suppose that $(-\partial\gamma,\gamma)$ and $(-\partial\gamma',\gamma')$ differ in a boundary at infinity and $(-\partial^T\delta,\delta)$ and $(-\partial^T\delta',\delta')$ differ in a coboundary at infinity. Then $\pi_{\omega_{(\gamma,\delta)}}$ is unitary equivalent to $\pi_{\omega_{(\gamma',\delta')}}$.
\end{prop}
\begin{proof}
If $\gamma-\gamma'$ is a locally finite boundary and $\delta-\delta'$ is a coboundary then $\omega_{(\gamma,\delta)}=\omega_{(\gamma',\delta')}$ by Proposition \ref{prop:sameexcitation}. Therefore we can assume that $|\supp(\gamma-\gamma')|<\infty$ and $|\supp(\delta-\delta')|<\infty$.

Since $\rho_{\gamma'}=\rho_{\gamma}\circ\rho_{\gamma'-\gamma}$ and $\rho_{\delta'}=\rho_{\delta}\circ\rho_{\delta'-\delta}$ we have
\begin{equation}
\omega_{(\gamma',\delta')}=\omega_{(\gamma,\delta)}\circ\rho_{(\gamma'-\gamma,\delta'-\delta)}.
\end{equation}
By the assumption on $\gamma-\gamma'$ and $\delta-\delta'$ the map $\rho_{(\gamma'-\gamma,\delta'-\delta)}$ is an inner automorphism, which implies that the GNS representations of the states $\omega_{(\gamma,\delta)}$ and $\omega_{(\gamma',\delta')}$ are unitary equivalent.
\end{proof}

Next we would like to find a sufficient condition for such states to have inequivalent GNS representations. We do this by exhibiting elements in the center of the von~Neumann algebra generated by the representations such that their expected value differs in the two states. For this construction we only need the weaker condition on the state that it resembles a frustration free ground state when we move away far enough towards infinity. Such states will be called asymptotic ground states.
\begin{defn}\label{def:astrong}
A state $\omega:\mathfrak{A}\to\mathbb{C}$ is an asymptotic ground state if for any $\varepsilon>0$ there exist finite subsets $K_{\pm}(\varepsilon)\subseteq\mathcal{E}_{n\pm 1}$ such that for any $c\in C_{\text{lf}}^{n-1}(E;G)$ and $d\in C_{n+1}(E;\hat{G})$ with $\supp c\cap K_-(\varepsilon)=\emptyset$ and $\supp d\cap K_+(\varepsilon)=\emptyset$ the inequality
\begin{equation}
|\omega(X^{\partial^T c}Z^{\partial d})-1|<\varepsilon
\end{equation}
holds.
\end{defn}

Clearly, the set of asymptotic ground states is convex, contains the set of frustration free ground states and is closed in the uniform topology.

\begin{thm}\label{thm:polarization}
Suppose that $\omega:\mathfrak{A}\to\mathbb{C}$ is an asymptotic ground state. Then
\begin{enumerate}
\item for any $c\in C^{n-1}(E;G)$ and $d\in C^{\text{lf}}_{n+1}(E;\hat{G})$ with $|\supp\partial^T c|<\infty$ and $|\supp\partial d|<\infty$ the limit
\begin{equation}
\lim_{\substack{K_-\to\mathcal{E}_{n-1}  \\  K_+\to\mathcal{E}_{n+1}}}\pi_{\omega}\left(X^{\partial^T(c-c_{K_-})}Z^{\partial(d-d_{K_+})}\right)
\end{equation}
exists in the strong operator topology and it is unitary
\item the limit is in $\pi_\omega(\mathfrak{A})'\cap\pi_\omega(\mathfrak{A})''$
\item the limit only depends on $\homatinfty{d}$ and $\cohomatinfty{c}$
\item the map $P_\omega:H^\infty_n(E;\hat{G})\times H_\infty^{n-1}(E;G)\to\pi_\omega(\mathfrak{A})'\cap\pi_\omega(\mathfrak{A})''\hookrightarrow\boundeds(\mathcal{H}_{\omega})$ defined by this limit is a group homomorphism into the subset of unitary elements (i.e. a unitary representation on $\mathcal{H}_\omega$).
\end{enumerate}
\end{thm}
\begin{proof}
Since the net in question is norm-bounded, it is enough to prove pointwise convergence in norm on a set spanning a dense subspace. Fix some $a\in C_{\text{lf}}^n(E;G)$ and $b\in C_n(E;\hat{G})$. For any $\varepsilon>0$ choose $K_\pm(\varepsilon)\subseteq\mathcal{E}_{n\pm 1}$ as in the definition, but with the additional property $\supp\partial^T a\subseteq K_+(\varepsilon)$ and $\supp\partial b\subseteq K_-(\varepsilon)$. Finally, let $K_\pm,K'_\pm\subseteq\mathcal{E}_{n\pm 1}$ be finite subsets such that $K_\pm(\varepsilon)\subseteq K_\pm$ and $K_\pm(\varepsilon)\subseteq K'_\pm$. Then
\begin{equation}
\begin{split}
& \norm{\pi_\omega(X^{\partial^T(c-c_{K_-})}Z^{\partial(d-d_{K_+})})[X^aZ^b]-\pi_\omega(X^{\partial^T(c-c_{K'_-})}Z^{\partial(d-d_{K'_+})})[X^aZ^b]}^2  \\
& = \norm{\pi_\omega(X^{\partial^T(c-c_{K_-})}Z^{\partial(d-d_{K_+})})\left(I-\pi_\omega(X^{\partial^T(c_{K_-}-c_{K'_-})}Z^{\partial(d_{K_+}-d_{K'_+})})\right)[X^aZ^b]}^2  \\
& = \omega\left(Z^{-b}X^{-a}\left(I-X^{\partial^T(c_{K_-}-c_{K'_-})}Z^{\partial(d_{K_+}-d_{K'_+})}\right)^*\left(I-X^{\partial^T(c_{K_-}-c_{K'_-})}Z^{\partial(d_{K_+}-d_{K'_+})}\right)X^aZ^b\right)  \\
& = 2-2\Re e^{2\pi i(\langle d_{K_+}-d_{K'_+},\partial^T a\rangle+\langle \partial b,c_{K_-}-c_{K'_-}\rangle)}\omega(X^{\partial^T(c_{K_-}-c_{K'_-})}Z^{\partial(d_{K_+}-d_{K'_+})})  \\
& = 2\left(1-\Re\omega\left(X^{\partial^T(c_{K_-}-c_{K'_-})}Z^{\partial(d_{K_+}-d_{K'_+})}\right)\right)<2\varepsilon,
\end{split}
\end{equation}
estabilishing that the net of such vectors is Cauchy.

Since the net goes in $\pi_\omega(\mathfrak{A})$ and the limit exists in the strong operator topology, it is contained in the strong closure $\pi_\omega(\mathfrak{A})''$. To see that it is also in the commutant it is enough to compute the strong limit of commutators with local algebra elements, because multiplication is continuous separately in its variables in the strong operator topology. Actually, even more is true: the net of such commutators is eventually $0$. To see this, let $a\in C_{\text{lf}}^n(E;G)$ and $b\in C_n(E;\hat{G})$. Then
\begin{equation}
\begin{split}
& \left[\pi_\omega\left(X^{\partial^T(c-c_{K_-})}Z^{\partial(d-d_{K_+})}\right),\pi_\omega(X^aZ^b)\right]  \\
& = \pi_\omega(X^{\partial^T(c-c_{K_-})}Z^{\partial(d-d_{K_+})}X^aZ^b-X^aZ^bX^{\partial^T(c-c_{K_-})}Z^{\partial(d-d_{K_+})})  \\
& = (1-e^{2\pi i(\langle b,\partial^T(c-c_{K_-})\rangle-\langle \partial(d-d_{K_+}),a\rangle)})\pi_\omega(X^{\partial^T(c-c_{K_-})}Z^{\partial(d-d_{K_+})}X^aZ^b)
\end{split}
\end{equation}
and here the first factor is $0$ if the sets $K_\pm$ contain $\supp\partial b$ and $\supp\partial^T a$.

The limit of a net of unitaries converging in the strong operator topology is an isometry. Taking the inverse of each element we get the net
\begin{equation}
\pi_{\omega}\left(X^{\partial^T((-c)-(-c)_{K_-})}Z^{\partial((-d)-(-d)_{K_+})}\right),
\end{equation}
which is of the same form, therefore it also converges strongly and we can conclude that the limit is unitary.

In the following part we prefer to work with weak convergence, using the fact that the weak and strong limits of a strongly convergent net agree.
If we add an $n+1$-chain to $d$ and a locally finite $n-1$-cochain to $c$ then for large enough $K_\pm$ the differences $d-d_{K_+}$ and $c-c_{K_-}$ remain unchanged. Suppose now that we add a locally finite $n+1$-boundary $\partial f$ to $d$ and an $n-1$-coboundary $\partial^Te$ to $c$. Then for any $a_1,a_2,b_1,b_2$ as before and for large enough $K_\pm$ we have
\begin{multline}\label{eq:boundaryadded}
\left\langle[X^{a_1}Z^{b_1}],\pi_\omega\left(X^{\partial^T((c+\partial^Te)-(c+\partial^Te)_{K_-})}Z^{\partial((d+\partial f)-(d+\partial f)_{K_+})}\right)[X^{a_2}Z^{b_2}]\right\rangle  \\
=\omega\left(Z^{-b_1}X^{-a_1}X^{a_2}Z^{b_2}X^{\partial^T((c+\partial^Te)-(c+\partial^Te)_{K_-})}Z^{\partial((d+\partial f)-(d+\partial f)_{K_+})}\right)  \\
=\omega\left(Z^{-b_1}X^{-a_1}X^{a_2}Z^{b_2}X^{\partial^T(c-c_{K_-})}Z^{\partial(d-d_{K_+})}X^{-\partial^T((\partial^Te)_{K_-})}Z^{-\partial((\partial f)_{K_+})}\right).
\end{multline}

\begin{figure}
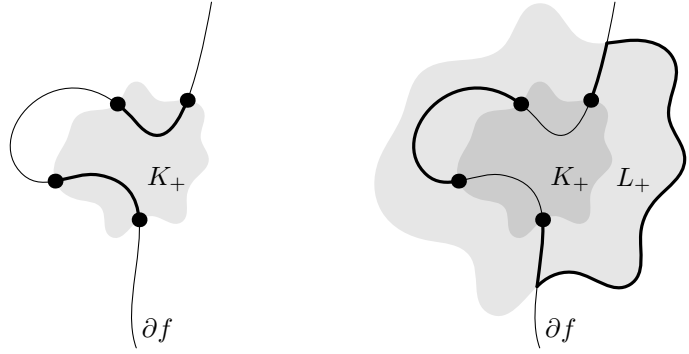

\begin{subfigure}{.45\textwidth}
\centering
\includegraphics{finiteboundary1.mps}
\end{subfigure}
\begin{subfigure}{.45\textwidth}
\centering
\includegraphics{finiteboundary2.mps}
\end{subfigure}
\caption{Left: If $f$ is a locally finite $n+2$-chain, we can take its boundary, truncate it to the finite set $K_+\subseteq\mathcal{E}_{n+1}$ then take the boundary again. Right: The same can be obtained as the boundary of an $n+1$-chain supported \emph{outside} $K_+$. To do this, first truncate $-f$ to a slightly larger set $L_+\subseteq\mathcal{E}_{n+2}$, take the boundary, truncate to the complement of $K_+$ and take the boundary again. ($n=0$)\label{fig:finiteboundary}}
\end{figure}
The key observation here is that $-\partial^T((\partial^Te)_{K_-})$ can also be written as the coboundary of a locally finite $n-1$-cochain supported outside $K_-$ and similarly, $-\partial((\partial f)_{K_+})$ can be written as the boundary of an $n+1$ chain supported outside $K_+$ (Figure \ref{fig:finiteboundary}). To see this, let $L_+=\bigcup_{\alpha\in K_+,g\in G}\supp(\partial^T(ge_{\alpha}))$, i.e. the set of those $n+2$-cells whose boundary intersects $K_+$. Then
\begin{equation}
\begin{split}
0
 & = \partial^2(f_{L_+})  \\
 & = \partial\left((\partial f_{L_+})_{K_+}+(\partial f_{L_+})_{\mathcal{E}_{n+1}\setminus K_+}\right)  \\
 & = \partial\left((\partial(f-f_{\mathcal{E}_{n+2}\setminus L_+}))_{K_+}\right)+\partial\left((\partial f_{L_+})_{\mathcal{E}_{n+1}\setminus K_+}\right)  \\
 & = \partial\left((\partial f)_{K_+}\right)-\partial\left((\partial f_{\mathcal{E}_{n+2}\setminus L_+})_{K_+}\right)+\partial\left((\partial f_{L_+})_{\mathcal{E}_{n+1}\setminus K_+}\right).
\end{split}
\end{equation}
Here the second term is $0$ because of the way $L_+$ was chosen and the last term is the boundary of an $n+1$-chain supported outside $K_+$, since $|L_+|<\infty$. The proof for $-\partial^T((\partial^Te)_{K_-})$ is similar with $L_-=\bigcup_{\beta\in K_-,\chi\in \hat{G}}\supp(\partial (\chi e_{\beta}))$. Since $\omega$ is an asymptotic ground state, for any $\varepsilon>0$ and for large enough $K_{\pm}$ we have
\begin{equation}
\left|\omega\left(X^{-\partial^T((\partial^Te)_{K_-})}Z^{-\partial((\partial f)_{K_+})}\right)-1\right|<\varepsilon,
\end{equation}
therefore by Lemma \ref{lem:unitary} its contribution in eq. \eqref{eq:boundaryadded} becomes negligible.

Finally, for any $d,d'\in C^{\text{lf}}_{n+1}(E;\hat{G})$ and $c,c'\in C^{n-1}(E;G)$ as before we have
\begin{equation}
\begin{split}
& X^{\partial^T(c-c_{K_-})}Z^{\partial(d-d_{K_+})}X^{\partial^T(c'-c'_{K_-})}Z^{\partial(d'-d'_{K_+})}  \\
 & = e^{2\pi i\langle\partial(d-d_{K_+}), \partial^T(c'-c'_{K_-})\rangle}X^{\partial^T((c+c')-(c+c')_{K_-})}Z^{\partial((d+d')-(d+d')_{K_+})}  \\
 & = e^{2\pi i\langle\partial^2(d-d_{K_+}), c'-c'_{K_-}\rangle}X^{\partial^T((c+c')-(c+c')_{K_-})}Z^{\partial((d+d')-(d+d')_{K_+})}  \\
 & = X^{\partial^T((c+c')-(c+c')_{K_-})}Z^{\partial((d+d')-(d+d')_{K_+})}.
\end{split}
\end{equation}
A net of unitaries is bounded and multiplication is jointly continuous in the strong operator topology when restricted to a norm-bounded set, so our map is a homomorphism as claimed.
\end{proof}

\begin{defn}\label{def:polarization}
We call the representation $P_\omega:H^\infty_n(E;\hat{G})\times H_\infty^{n-1}(E;G)\to\boundeds(\mathcal{H}_{\omega})$ in the theorem above the polarization of $\omega$.
\end{defn}

As an example, take $n=0$, $G=\mathbb{Z}_2$ (i.e. the ferromagnetic Ising model) on the half-line. Then there are two frustration free ground states with two different magnetic polarizations and the value of $P_\omega$ on the nontrivial element is $\pm1$ times the identity in the corresponding representations.

\begin{prop}\label{prop:polarization-endomorphism}
Let $\omega:\mathfrak{A}\to\mathbb{C}$ be an asymptotic ground state and let $\gamma\in C^{\text{lf}}_n(E;\hat{G})$ and $\delta\in C^n(E;G)$ with $|\supp\partial\gamma|<\infty$ and $|\supp\partial^T\delta|<\infty$. Then $\omega':=\omega\circ\rho_{(\gamma,\delta)}$ is also an asymptotic ground state. Moreover, the representation $P_{\omega'}:H^\infty_n(E;\hat{G})\times H_\infty^{n-1}(E;G)\to \boundeds(\mathcal{H}_{\omega'})$ is isomorphic to $P_\omega\otimes V_{(\gamma,\delta)}$ where $V_{(\gamma,\delta)}$ is a one dimensional representation where the pair $(\homatinfty{d},\cohomatinfty{c})$ acts as multiplication by
\begin{equation}
e^{2\pi i(\langle\homatinfty{\gamma},\cohomatinfty{c}\rangle+\langle\homatinfty{d},\cohomatinfty{\delta}\rangle)}.
\end{equation}
\end{prop}
\begin{proof}
For simplicity, in this proof we identify the representation space of $V_{(\gamma,\delta)}$ with $\mathbb{C}$ and use $\mathcal{H}_\omega\otimes\mathbb{C}=\mathcal{H}_\omega$ (as Hilbert spaces).

We will show that $U[A]_\omega=[\rho_{-\gamma}\rho_{-\delta}(A)]_{\omega'}$ provides the required isomorphism. Since
\begin{equation}
\begin{split}
\norm{[\rho_{-\gamma}\rho_{-\delta}(A)]_{\omega'}}^2
 & = \omega'((\rho_{-\gamma}\rho_{-\delta}(A))^*\rho_{-\gamma}\rho_{-\delta}(A))  \\
 & = \omega'(\rho_{-\gamma}\rho_{-\delta}(A^*A))  \\
 & = \omega(A^*A)=\norm{[A]_\omega}^2
\end{split}
\end{equation}
for any $A$ we have that $J_{\omega'}=\rho_{-\gamma}\rho_{-\delta}(J_\omega)$, therefore the map is well-defined and isometric. For the same reason, $[A]_{\omega'}\mapsto[\rho_{\gamma}\rho_{\delta}A]_\omega$ is well-defined, and as this is clearly $U^{-1}$ we have that $U$ is unitary.

To see that $U$ is equivariant, take any $A,B\in\mathfrak{A}$, $d\in C^{\text{lf}}_{n+1}(E;\hat{G})$, $c\in C^{n-1}(E;G)$ with $\supp\partial d$ and $\supp\partial^T c$ finite. Let us abbreviate $P_\omega(\homatinfty{d},\cohomatinfty{c})$ by $P_\omega$ and similarly for $P_{\omega'}$. Then
\begin{equation}
\begin{split}
 & \langle U[B]_\omega,P_{\omega'} U[A]_\omega\rangle  \\
 & = \langle [\rho_{-\gamma}\rho_{-\delta}(B)]_{\omega'},P_{\omega'} [\rho_{-\gamma}\rho_{-\delta}(A)]_{\omega'}\rangle  \\
 & = \lim_{K\pm\to\mathcal{E}_{n\pm 1}}\langle [\rho_{-\gamma}\rho_{-\delta}(B)]_{\omega'},\pi_{\omega'}(X^{\partial^T(c-c_{K_-})}Z^{\partial(d-d_{K_+})})[\rho_{-\gamma}\rho_{-\delta}(A)]_{\omega'}\rangle  \\
 & = \lim_{K\pm\to\mathcal{E}_{n\pm 1}}\omega'(\rho_{-\gamma}\rho_{-\delta}(B^*)X^{\partial^T(c-c_{K_-})}Z^{\partial(d-d_{K_+})}\rho_{-\gamma}\rho_{-\delta}(A))  \\
 & = \lim_{K\pm\to\mathcal{E}_{n\pm 1}}\omega(B^*\rho_{\gamma}\rho_{\delta}(X^{\partial^T(c-c_{K_-})}Z^{\partial(d-d_{K_+})})A)  \\
 & = \lim_{K\pm\to\mathcal{E}_{n\pm 1}}e^{2\pi i(\langle\gamma,\partial^T(c-c_{K_-})\rangle-\langle\partial(d-d_{K_+}),\delta\rangle)}\omega(B^*X^{\partial^T(c-c_{K_-})}Z^{\partial(d-d_{K_+})}A)  \\
 & = \lim_{K\pm\to\mathcal{E}_{n\pm 1}}e^{2\pi i(\langle\gamma,\partial^Tc\rangle-\langle\partial\gamma,c_{K_-}\rangle-\langle\partial d,\delta\rangle+\langle d_{K_+},\partial^T\delta\rangle)}\langle[B],\pi_\omega(X^{\partial^T(c-c_{K_-})}Z^{\partial(d-d_{K_+})})[A]\rangle  \\
 & = e^{2\pi i(\langle\gamma,\partial^Tc\rangle-\langle\partial\gamma,c\rangle+\langle d,\partial^T\delta\rangle-\langle\partial d,\delta\rangle)}\langle[B],P_\omega[A]\rangle
\end{split}
\end{equation}
holds, i.e. $U^*P_{\omega'}U=e^{2\pi i(\langle\homatinfty{\gamma},\cohomatinfty{c}\rangle+\langle\homatinfty{d},\cohomatinfty{\delta}\rangle)}P_\omega$.
\end{proof}

Since a frustration free ground state $\omega_0$ is an asymptotic ground state, the same is true for states $\omega_0\circ\rho_{(\gamma,\delta)}$ with $|\supp\partial\gamma|<\infty$ and ${|\supp\partial^T\delta|<\infty}$. From Proposition \ref{thm:equivalent} and Theorem \ref{thm:polarization} it follows that if the pairings $H^\infty_{n-1}(E;\hat{G})\times H_\infty^{n-1}(E;G)\to\mathbb{Q}/\mathbb{Z}$ and $H^\infty_n(E;\hat{G})\times H_\infty^n(E;G)\to\mathbb{Q}/\mathbb{Z}$ are nondegenerate, then quasi-equivalence classes of such states are distinguished by their polarizations. This is the case e.g. for spaces of the form $F\times\mathbb{R}^k$ with $F$ a compact CW complex as well as for spaces obtained by changing a finite part in such a space.

In the toric code the GNS representations under consideration satisfy an important selection criterion. Namely, they are unitary equivalent to the frustration free ground state (vacuum) when restricted to observables supported in the complement of any infinite cone. In our model there seem to be no obvious analogues of cones, but a similar property can be derived for certain regions. However, the allowed regions may depend on the representation.
\begin{prop}
Let $\omega_{(\gamma,\delta)}$ be a state as before with $|\supp\partial\gamma|<\infty$ and $|\supp\partial^T\delta|<\infty$. Suppose that $C\subseteq\mathcal{E}_n$ is a region such that the (co-)homology classes at infinity $\homatinfty{\gamma}$ and $\cohomatinfty{\delta}$ have representatives $(-\partial\gamma',\gamma')$ and $(-\partial\delta',\delta')$ with $(\supp\gamma')\cup(\supp\delta')\subseteq C$. Then
\begin{equation}
\pi_{\omega_{(\gamma,\delta)}}|_{\mathfrak{A}_{\mathcal{E}_n\setminus C}}\simeq\pi_{\omega_0}|_{\mathfrak{A}_{\mathcal{E}_n\setminus C}}.
\end{equation}
\end{prop}
\begin{proof}
Choose $\gamma'$ and $\delta'$ as in the condition. Then by Proposition \ref{thm:equivalent} we have $\pi_{\omega_{(\gamma,\delta)}}\simeq\pi_{\omega_{(\gamma',\delta')}}$. In general, it is true that $\pi_\omega\circ\rho$ is a GNS representation for the state $\omega\circ\rho$, hence by uniqueness, $\pi_{\omega_0}\circ\rho_{(\gamma',\delta')}$ is also unitarily equivalent to $\pi_{\omega_{(\gamma,\delta)}}$. But on $\mathfrak{A}_{\mathcal{E}_n\setminus C}$ the endomorphism $\rho_{(\gamma',\delta')}$ acts as the identity map, so restricted to this subalgebra the representation is the same as $\pi_{\omega_0}$.
\end{proof}

The statement contains as a special case the corresponding one for cones in the plane. This is because \emph{any} element of $H_\infty^n(\mathbb{R}^2,G)$ and $H^\infty_{n-1}(\mathbb{R}^2,\hat{G})$ has a representative supported in any infinite cone region.

This proposition suggests that the role of cones should be played by regions where we can find representatives of (co-)homology classes at infinity. However, the exact condition depends on the equivalence class of the representation. One might try to impose the condition that \emph{every} (co-)homology class at infinity be represented within a region if it is to be admissible, but this idea does not seem to work. The problem is that in some cases this forces the complement to be too small in the sense that too many representations become equivalent to the vacuum when restricted to it. This is the case e.g when $E=\mathbb{R}$ and $n=0$, where such regions are the cofinite sets of $0$-cells, therefore any two representations become equivalent when restricted to their complements.

\section{Braiding}\label{sec:braiding}

In this section $\omega_0$ will be a fixed pure frustration free ground state, $\pi_0:\mathfrak{A}\to\boundeds(\mathcal{H})$ its GNS representation and for any endomorphism $\rho:\mathfrak{A}\to\mathfrak{A}$ we will identify the GNS representation of $\omega_0\circ\rho$ with $\pi_0\circ\rho$. Under this identification the distinguished cyclic vectors coincide. This vector will be denoted by $\Omega_0$.

In the DHR analysis one defines a category with localized and transportable endomorphisms as objects and intertwiners as morphisms (see e.g. ref. \cite{Halvorson}). In that case such an intertwiner $T\in\Hom(\rho,\rho')$ is an element of $\mathfrak{A}$ such that $T\rho(A)=\rho'(A)T$ for every $A\in\mathfrak{A}$. Applying $\pi_0$ to both sides one gets that $\pi_0(T)$ is an intertwiner between the corresponding representations. This category comes equipped with a tensor product defined as $\rho\otimes\sigma=\rho\circ\sigma$ for objects and as $S\otimes T=S\rho(T)$ for morphisms $S\in\Hom(\rho,\rho')$ and $T\in\Hom(\sigma,\sigma')$. Then a braiding can be defined as $\varepsilon_{\rho_1,\rho_2}(U_1,U_2)=\rho_2(U_1^*)U_2^*U_1\rho_1(U_2)\in\Hom(\rho_1\otimes\rho_2,\rho_2\otimes\rho_1)$ where $U_i\in\Hom(\rho_i,\rho_i')$ are unitary intertwiners for some ``spectator morphisms'' $\rho_i'$. On $\mathbb{R}^d$ with $d\ge 3$ this definition actually does not depend on the spectator morphisms and intertwiners chosen.

However, in our case (and also in the planar Kitaev model) the objects are not localized in compact regions and consequently, intertwiners are not in $\pi_0(\mathfrak{A})$, but rather in the von~Neumann algebra generated by it. Hence it is not possible to take their images under an endomorphism of $\mathfrak{A}$. One possibility to circumvent this problem is to show that the endomorphisms extend to some algebra containing the intertwiners. This route is taken in refs. \cite{NaaijkensLocalized,NaaijkensDuality} for the toric code on the plane, showing that localized endomorphisms extend to the von~Neumann algebra generated by local observables supported outside some fixed cone.

Our approach will be slightly different, since it is not clear what kind of region should be excluded. Still, the calculations will be very similar. The idea is that given a pair of locally finite $n$-chains $\gamma,\gamma'$ with finite boundaries and such that $\gamma'=\gamma-\hat{\gamma}+\partial p$ where $p$ is a locally finite $n+1$-chain and $|\supp\hat{\gamma}|<\infty$ it is possible to construct a net of unitaries in $\mathfrak{A}$ converging to a unitary intertwiner $U\in\Hom(\rho_\gamma,\rho_{\gamma'})$. Moreover, this can be done in such a way that applying an endomorphism of similar form to each element in the net results in another convergent net. This is the content of the following propositions.
\begin{figure}
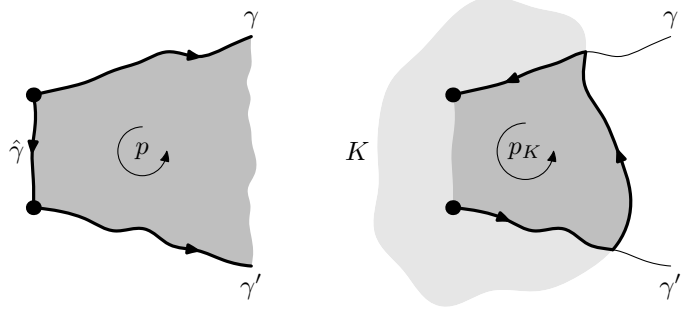

\begin{subfigure}{.45\textwidth}
\centering
\includegraphics{transporter1.mps}
\end{subfigure}
\begin{subfigure}{.45\textwidth}
\centering
\includegraphics{transporter2.mps}
\end{subfigure}
\caption{Left: $\gamma$ and $\gamma'$ represent the same class in $H^\infty_{n-1}(E;\hat{G})$, therefore it is possible to find $\hat{\gamma}\in C_n(E;\hat{G})$ and $p\in C^{\text{lf}}_{n+1}(E;\hat{G})$ such that $\gamma'=\gamma-\hat{\gamma}+\partial p$. Right: For each finite region $K\subseteq\mathcal{E}_{n+1}$ we can truncate $p$ to get an $n+1$ chain $p_K$ and form the operator corresponding to the $n$-chain $-\hat{\gamma}+\partial(p_K)$. As $K$ grows, the operators converge to a unitary intertwiner. ($n=1$)\label{fig:transporter}}
\end{figure}
\begin{prop}\label{prop:transporter}
Let $\gamma,\gamma'\in C^{\text{lf}}_n(E;\hat{G})$ and $\delta,\delta'\in C^n(E;G)$ such that $\homatinfty{\gamma}=\homatinfty{\gamma'}\in H^\infty_{n-1}(E;\hat{G})$ and $\cohomatinfty{\delta}=\cohomatinfty{\delta'}\in H_\infty^n(E;G)$. Choose $\hat{\gamma}\in C_n(E;\hat{G})$, $\hat{\delta}\in C_{\text{lf}}^n(E;G)$, $p\in C^{\text{lf}}_{n+1}(E;\hat{G})$ and $q\in C^{n-1}(E;G)$ such that
\begin{subequations}\label{eq:transportdata}
\begin{align}
\gamma' & = \gamma-\hat{\gamma}+\partial p  \\
\delta' & = \delta-\hat{\delta}+\partial^T q.
\end{align}
\end{subequations}
For $K_\pm\subseteq\mathcal{E}_{n\pm 1}$ let
\begin{equation}
U_{K_+,K_-}=Z^{-\hat{\gamma}}X^{-\hat{\delta}}X^{\partial^T(q_{K_-})}Z^{\partial(p_{K_+})}.
\end{equation}
(The construction is illustrated in Figure \ref{fig:transporter}.) Then the limit
\begin{equation}
U:=\lim_{K_\pm\to\mathcal{E}_{n\pm 1}}\pi_0(U_{K_+,K_-})
\end{equation}
exists in the strong operator topology, satisfies $U\pi_0(\rho_{(\gamma,\delta)}(A))=\pi_0(\rho_{(\gamma',\delta')}(A))U$ for any $A\in\mathfrak{A}$ (i.e. it is an intertwiner for the two representations). Moreover, $\pi_0(X^{\hat{\delta}}Z^{\hat{\gamma}})U\Omega_0=\Omega_0$.
\end{prop}
\begin{proof}
The net is bounded, therefore it is enough to verify convergence on vectors of the form $[X^aZ^b]$. If $K_+\supseteq\supp\partial^Ta$ and $K_-\supseteq\supp\partial b$ then we have
\begin{equation}
\begin{split}
\pi_0(U_{K_+,K_-})[X^aZ^b]
 & = [Z^{-\hat{\gamma}}X^{-\hat{\delta}}X^{\partial^T(q_{K_-})}Z^{\partial(p_{K_+})}X^aZ^b]  \\
 & = e^{2\pi i(\langle\partial(p_{K_+}),a\rangle-\langle b,\partial^T(q_{K_-})\rangle)}[Z^{-\hat{\gamma}}X^{-\hat{\delta}}X^aZ^bX^{\partial^T(q_{K_-})}Z^{\partial(p_{K_+})}]  \\
 & = e^{2\pi i(\langle p_{K_+},\partial^Ta\rangle-\langle \partial b,q_{K_-}\rangle)}[Z^{-\hat{\gamma}}X^{-\hat{\delta}}X^aZ^b]  \\
 & = e^{2\pi i(\langle p,\partial^Ta\rangle-\langle \partial b,q\rangle)}[Z^{-\hat{\gamma}}X^{-\hat{\delta}}X^aZ^b],
\end{split}
\end{equation}
which no longer depends on $K_\pm$, therefore this is also how the strong limit acts.

To see that $U$ is an intertwiner, we use strong continuity of multiplication in bounded sets of operators:
\begin{equation}
\begin{split}
U\pi_0(X^aZ^b)U^{-1}
 & = \lim_{K\pm\to\mathcal{E}_{n\pm1}}\pi_0(U_{K_+,K_-}X^aZ^bU_{K_+,K_-}^*)  \\
 & = \lim_{K\pm\to\mathcal{E}_{n\pm1}}\pi_0(Z^{-\hat{\gamma}}X^{-\hat{\delta}}X^{\partial^T(q_{K_-})}Z^{\partial(p_{K_+})}X^aZ^bZ^{-\partial(p_{K_+})}X^{-\partial^T(q_{K_-})}X^{\hat{\delta}}Z^{\hat{\gamma}})  \\
 & = \lim_{K\pm\to\mathcal{E}_{n\pm1}}e^{2\pi i(\langle-\hat{\gamma}+\partial(p_{K_+}),a\rangle-\langle b,-\hat{\delta}+\partial^T(q_{K_-})\rangle)}\pi_0(X^aZ^b)  \\
 & = e^{2\pi i(\langle\gamma'-\gamma,a\rangle-\langle b,\delta'-\delta\rangle)}\pi_0(X^aZ^b)  \\
 & = (\pi_0\circ\rho_{(\gamma',\delta')}\circ\rho_{(\gamma,\delta)}^{-1})(X^aZ^b)
\end{split}
\end{equation}
by eqs. \eqref{eq:transportdata} and \eqref{eq:transformedstate}.

Finally,
\begin{equation}
\begin{split}
\pi_0(X^{\hat{\delta}}Z^{\hat{\gamma}})U\Omega_0
 & = \lim_{K\pm\mathcal{E}_{n\pm1}}[X^{\hat{\delta}}Z^{\hat{\gamma}}Z^{-\hat{\gamma}}X^{-\hat{\delta}}X^{\partial^T(q_{K_-})}Z^{\partial(p_{K_+})}]  \\
 & = \lim_{K\pm\mathcal{E}_{n\pm1}}[I]=\Omega_0
\end{split}
\end{equation}
proves the last statement.
\end{proof}

As Figure \ref{fig:transporter} shows, the string operators in the converging net can be visualized as curves that start at the endpoint of the second semi-infinite line, go along this line for a while, then connect the two lines with a ``distant'' segment and go back along the first line. As $K$ grows, the contribution of the connecting segment eventually commutes with any fixed local observable. This picture explains in a heuristic manner why is it not possible to find charge transporters when $\homatinfty{\gamma}\neq\homatinfty{\gamma'}$. As a simple example we may imagine a region in the plane which is the union of several overlapping semi-infinite bands (see Figure \ref{fig:notransport}). In this case one can find a compact region which cannot be avoided by any curve joining distant points of distinct bands.
\begin{figure}
\centering
\includegraphics{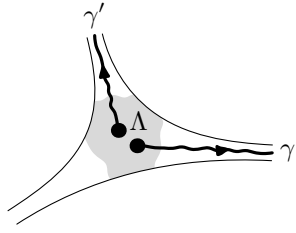}
\caption{For a CW complex with nontrivial end space it is not always possible to construct charge transporters for two states which look similar locally. Any path joining distant points of the two strings $\gamma$ and $\gamma'$ must cross the compact region $\Lambda$. ($n=1$)\label{fig:notransport}}
\end{figure}

Let us see what happens if we apply an endomorphism to each element in the net. Let $\gamma$ and $\delta$ be and arbitrary locally finite $n$-chain and $n$-cochain with finite boundary and coboundary, respectively. Then
\begin{equation}
\begin{split}
\pi_0(\rho_{(\gamma,\delta)}(U_{K_+,K_-}))
 & = e^{2\pi i(\langle\gamma,-\hat{\delta}+\partial^T(q_{K_-})\rangle-\langle-\hat{\gamma}+\partial(p_{K_+}),\delta\rangle)}\pi_0(U_{K_+,K_-})  \\
 & = e^{2\pi i(\langle\hat{\gamma},\delta\rangle-\langle\gamma,\hat{\delta}\rangle+\langle\partial\gamma,q_{K_-}\rangle-\langle p_{K_+},\partial^T\delta\rangle)}\pi_0(U_{K_+,K_-}).
\end{split}
\end{equation}
The first factor is eventually constant, while the second one has a limit in the strong operator topology, therefore the product also has a limit and is equal to
\begin{equation}
e^{2\pi i(\langle\hat{\gamma},\delta\rangle-\langle\gamma,\hat{\delta}\rangle+\langle\partial\gamma,q\rangle-\langle p,\partial^T\delta\rangle)}U.
\end{equation}
Now we are in a position to compute the braiding morphisms.
\begin{prop}\label{prop:braiding}
For $i=1,2$ let $\gamma_i\in C^{\text{lf}}_n(E;\hat{G})$ and $\delta_i\in C^n(E;G)$ with finite boundaries and coboundaries, respectively. Let $\hat{\gamma}_i\in C_n(E;\hat{G})$, $\hat{\delta}_i\in C_{\text{lf}}^n(E;G)$, $p_i\in C^{\text{lf}}_{n+1}(E;\hat{G})$ and $q_i\in C^{n-1}(E;G)$. Take the unitaries
\begin{equation}
U_{i,K_+,K_-}=Z^{-\hat{\gamma}_i}X^{-\hat{\delta}_i}X^{\partial^T(q_{i,K_-})}Z^{\partial(p_{i,K_+})}
\end{equation}
as before and set $U_i=\lim_{K_{\pm}\to\mathcal{E}_{n\pm 1}}\pi_0(U_{i,K_+,K_-})$. Then the limit
\begin{multline}\label{eq:braidingdef}
\varepsilon_{(\gamma_1,\delta_1),(\gamma_2,\delta_2)}((\hat{\gamma}_1,\hat{\delta}_1,p_1,q_1),(\hat{\gamma}_2,\hat{\delta}_2,p_2,q_2))  \\  :=\lim_{K_{\pm}\to\mathcal{E}_{n\pm 1}}\pi_0(\rho_{(\gamma_2,\delta_2)}(U_{1,K_+,K_-}^*))U_2^*U_1\pi_0(\rho_{(\gamma_1,\delta_1)}(U_{2,K_+,K_-}))
\end{multline}
exists in the strong operator topology and is equal to
\begin{equation}\label{eq:braidinggeneral}
\pi_0(I)e^{2\pi i(\langle\gamma_2,\hat{\delta}_1\rangle+\langle\hat{\gamma}_2,\delta_1\rangle-\langle\hat{\gamma}_2,\hat{\delta}_1\rangle+\langle\partial\hat{\gamma}_2,q_1\rangle-\langle\partial\gamma_2,q_1\rangle+\langle p_2,\partial^T\hat{\delta}_1\rangle-\langle p_2,\partial^T\delta_1\rangle)-(\ldots)},
\end{equation}
where the omitted terms are obtained by interchanging $1$ and $2$.
\end{prop}
\begin{proof}
$\rho_{\gamma_i,\delta_i}$ acts on $U_{j,K_+,K_-}$ as multiplication by
\begin{multline}
e^{2\pi i(\langle\gamma_i,-\hat{\delta}_j+\partial^T(q_{j,K_-})\rangle-\langle-\hat{\gamma}_j+\partial(p_{j,K_+}),\delta_i\rangle)}  \\
= e^{2\pi i(-\langle\gamma_i,\hat{\delta}_j\rangle+\langle\partial\gamma_i,q_{j,K_-}\rangle+\langle\hat{\gamma}_j,\delta_i\rangle-\langle p_{j,K_+},\partial^T\delta_i\rangle)},
\end{multline}
therefore we need to consider two factors of this type (in the $i=2$, $j=1$ case the exponent should be multiplied by $-1$) and one coming from the (multiplicative) commutator of the $U$s. The latter is a multiple of $I$ and the coefficient looks like
\begin{multline}
e^{2\pi i\langle\hat{\gamma}_2-\partial(p_{2,K_+}),-\hat{\delta}_1+\partial^T(q_{1,K_-})\rangle-2\pi i\langle-\hat{\gamma}_1+\partial(p_{1,K_+}),\hat{\delta}_2-\partial^T(q_{2,K_-})\rangle}  \\
= e^{2\pi i(-\langle\hat{\gamma}_2,\hat{\delta}_1\rangle+\langle\partial\hat{\gamma}_2,q_{1,K_-}\rangle+\langle p_{2,K_+},\partial^T\hat{\delta}_1\rangle)-2\pi i(-\langle\hat{\gamma}_1,\hat{\delta}_2\rangle+\langle\partial\hat{\gamma}_1,q_{2,K_-}\rangle+\langle p_{1,K_+},\partial^T\hat{\delta}_2\rangle)}.
\end{multline}
Multiplying these together and taking the limit as $K_\pm\to\mathcal{E}_{n\pm1}$ gives eq. \ref{eq:braidinggeneral}.
\end{proof}

Under certain conditions an alternative formula for the braiding operators can be derived using the eqs. \eqref{eq:transportdata}. When the sets $\supp\gamma_i\cap\supp\delta_j$,  $\supp\gamma'_i\cap\supp\delta'_j$,  $\supp\partial p_i\cap\supp\delta'_j$ and  $\supp\gamma'_i\cap\supp\partial^Tq_j$ are finite, we can also write it as
\begin{equation}
\pi_0(I)e^{2\pi i[(\langle\gamma_2,\delta_1\rangle-\langle\gamma'_2,\delta'_1\rangle)+(\langle\partial p_2,\delta'_1\rangle-\langle p_2,\partial^T\delta'_1\rangle)+(\langle\gamma'_2,\partial^Tq_1\rangle-\langle\partial\gamma'_2,q_1\rangle)]-2\pi i[\ldots]},
\end{equation}
where the omitted terms are again obtained by interchanging $1$ and $2$.

From eq. \eqref{eq:braidinggeneral} it is clear that the braiding isomorphisms depend on the choice of $p_i$, $q_i$, $\hat{\gamma}_i$ and $\hat{\delta}_i$ even with fixed $\gamma'_i$ and $\delta'_i$. On the other hand, one expects that similarly to 2D systems this dependence is ``topological'' and not sensitive to small perturbations. This is made precise in the following proposition.
\begin{prop}\label{prop:freedominbraiding}
Let $\gamma_i,\delta_i,\hat{\gamma}_i,\hat{\delta}_i,p_i,q_i$ be as before. Then the braiding operator $\varepsilon_{(\gamma_1,\delta_1),(\gamma_2,\delta_2)}((\hat{\gamma}_1,\hat{\delta}_1,p_1,q_1),(\hat{\gamma}_2,\hat{\delta}_2,p_2,q_2))$ is left unchanged under any of the following transformations:
\begin{enumerate}
\item $p_i$ is replaced with $p_i+\hat{p}_i+\partial e_i$ and $\hat{\gamma}_i$ is replaced with $\hat{\gamma}_i+\partial\hat{p}_i$ where $\hat{p}_i\in C_{n+1}(E;\hat{G})$ and $e_i\in C^{\text{lf}}_{n+2}(E;\hat{G})$
\item $q_i$ is replaced with $q_i+\hat{q}_i+\partial^T f_i$ and $\hat{\delta}_i$ is replaced with $\hat{\delta}_i+\partial^T\hat{q}_i$ where $\hat{q}_i\in C_{\text{lf}}^{n-1}(E;G)$ and $f_i\in C_{\text{lf}}^{n-2}(E;G)$
\item $p_1$ is replaced with $p_1+p'$ and $\hat{\gamma}_1$ is replaced with $\hat{\gamma}_1+\hat{\gamma}'$ where $(\supp p')\cap(\partial^T(\hat{\delta}_2-\delta_2))=\emptyset$ and $\supp\hat{\gamma}'\cap\supp(\delta_2-\hat{\delta}_2+\partial^Tq_2)=\emptyset$
\item $q_1$ is replaced with $q_1+q'$ and $\hat{\delta}_1$ is replaced with $\hat{\delta}_1+\hat{\delta}'$ where $(\supp q')\cap(\supp \partial(\hat{\gamma}_2-\gamma_2))=\emptyset$ and $\supp\hat{\gamma}'\cap\supp(\gamma_2-\hat{\gamma}_2+\partial p_2)=\emptyset$
\end{enumerate}
and similarly with the roles of $1$ and $2$ interchanged.
\end{prop}
\begin{proof}
If $p_i$ is replaced with $p_i+\hat{p}_i$ and $\hat{\gamma}_i$ with $\hat{\gamma}_i+\partial\hat{p}_i$ then for any $K_+\supseteq\supp\hat{p}_i$ we have
\begin{equation}
\begin{split}
U_{i,K_+,K_-}
 & = Z^{-(\hat{\gamma}_i+\partial\hat{p}_i)}X^{-\hat{\delta}_i}X^{\partial^T(q_{i,K_-})}Z^{\partial(p_{i,K_+}+\hat{p}_i)}  \\
 & = e^{2\pi\langle\partial\hat{p}_i,\hat{\delta}_i\rangle}Z^{-\hat{\gamma}_i}X^{-\hat{\delta}_i}X^{\partial^T(q_{i,K_-})}Z^{\partial p_{i,K_+}}.
\end{split}
\end{equation}
The phase factors from $U_{i,K_+,K_-}$ and $U_{i,K_+,K_-}^*$ cancel in eq. \eqref{eq:braidingdef}, therefore the braiding remains the same.

From eq. \eqref{eq:braidinggeneral} one can see that $p_i$ only appears as $\langle p_i,\partial^T(\hat{\delta}_j-\delta_j),\rangle$. Adding a boundary $\partial e_i$ results in an additional term $\langle \partial e_i,\partial^T(\hat{\delta}_j-\delta_j),\rangle=\langle e_i,(\partial^T)^2(\hat{\delta}_j-\delta_j),\rangle=0$.

Again, since $p_1$ only appears as $\langle p_1,\partial^T(\hat{\delta}_2-\delta_2),\rangle$, it is clear that $p_1+p'$ with $\supp p'$ disjoint from that of $\partial^T(\hat{\delta}_2-\delta_2)$ leads to the same braiding.

Similarly, $\hat{\gamma}_1$ only appears as $\langle\hat{\gamma}_1,\delta_2-\hat{\delta}_2+\partial^Tq_2\rangle$, therefore adding $\hat{\gamma}'$ does not change the value of the bilinear pairing as long as its support is disjoint from that of $\delta_2-\hat{\delta}_2+\partial^Tq_2$.

The proofs for cochains are similar.
\end{proof}

Although there is some freedom in choosing the additional data required to define a braiding, it can still lead to different operators. In general, we do not know whether a canonical choice exists. However, in the plane and more generally, for spaces which are ``essentially planar'' one can find such a choice.

For definiteness, let us consider $E=F\times\mathbb{R}^2$ with $F$ compact, but it seems likely that the argument applies to more general spaces. In this case one can find excitations localized in cones, i.e. preimages of infinite circular sectors in $\mathbb{R}^2$ under the canonical projection $E\to\mathbb{R}^2$. After fixing a forbidden cone $C$ and an orientation of $\mathbb{R}^2$ one can introduce an ordering on the set of cones disjoint from $C$, saying that $C_1 < C_2$ if $C_1$ can be rotated to $C_2$ in the positive direction in such a way that it never intersects $C$.

Using this ordering one can fix the choices in the braiding as follows (see Figure \ref{fig:braiding}). Choose $\rho_{(\gamma_1,\delta_1)}$ and $\rho_{(\gamma_2,\delta_2)}$ localized in cones $C_1$ and $C_2$, both disjoint from $C$ (but possibly $C_1\cap C_2\neq\emptyset$). Let $\hat{\gamma}_1=\hat{\delta}_1=p_1=q_1=0$ and choose $\rho_{(\gamma'_2,\delta'_2)}$ such that it is localized in a cone $C_2'$, disjoint from $C_1$ and $C$ and with the additional requirement $C_1 < C_2'$. Then we can choose $\hat{\gamma}_2,\hat{\delta_2}$ to be disjoint from $\gamma_1,\delta_1$ and $p_2,q_2$ to be supported between $C_2$ and $C_2'$ and outside $C$.

\begin{figure}
\centering
\includegraphics{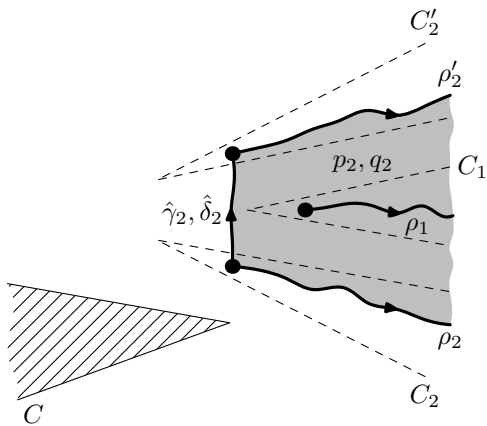}
\caption{In the plane one can fix a forbidden cone $C$ and consider endomorphisms and charge transporters supported outside $C$. Then a braiding can be defined unambiguously by requiring that the first charge stays in place while the second one is moved so that $C_1-C_2'-C$ are in counterclockwise order. ($n=1$)\label{fig:braiding}}
\end{figure}

Using Proposition \ref{prop:freedominbraiding} it is possible to give a different interpretation to the braiding operation. Let us start with an arrangement as in Figure \ref{fig:braiding}. Since the supports of $\hat{\gamma}_2$ and $\hat{\delta}_2$ are disjoint from those of $\delta_1$ and $\gamma_1$, we can replace them with $0$, thus changing $\gamma'_2$ to $\gamma'_2+\hat{\gamma}_2$ and $\delta'_2$ to $\delta'_2+\hat{\delta}_2$. After this transformation $\gamma'_2-\gamma_2$ is a locally finite boundary and $\delta'_2-\delta_2$ is a coboundary, and it is possible to choose $p'$ and $q'$ such that their supports are disjoint from those of $\partial^T\delta_1$ and $\partial\gamma_1$ (but they \emph{do} intersect $C$). By Proposition \ref{prop:freedominbraiding} we can replace $p_2$ with $p_2+p'$ and $q_2$ with $q_2+q'$. The net effect is that now $\partial p_2=\partial^T q_2=\hat{\gamma}_2=\hat{\delta}_2=0$ and we also have $\gamma_2=\gamma_2'$ and $\delta_2=\delta_2'$. The physical picture one may have in mind is that the endpoint of either $\gamma_2$ or $\delta_2$ (or both) which is at infinty is moved around the entire plane, making a full circle and returning back to the same position. Note that when $C_2<C_1<C'_2$ (as in the picture) then the canonical choice corresponds here to a full counterclockwise circle, whereas when $C_1<C_2$, we can take $p_2=q_2=0$. In the special case when $(\gamma_1,\delta_1)=(\gamma_2,\delta_2)$ none of these relations hold and we need to take either $p_2=0$ or $q_2=0$.

Now we can go one step further by noticing that we can add any $n+1$-chain to $p_2$ and locally finite $n+1$-chain to $q_2$ without affecting the braiding. What actually matter are the pairs $(p_2,\partial p_2)$ and $(q_2,\partial^Tq_2)$. Moreover, by Proposition \ref{prop:freedominbraiding} we can add any locally finite boundary to $p_2$ and any coboundary to $q_2$, i.e. the braiding only depends on the (co-)homology classes at infinity $\homatinfty{p_2}\in H^\infty_n(E;\hat{G})$ and $\cohomatinfty{q_2}\in H_\infty^{n-1}(E;G)$. According to this picture one may redefine the braiding operators to depend only on these (co-)homology classes at infinity:
\begin{equation}\label{eq:braidingatinfinity}
\varepsilon_{(\gamma_1,\delta_1),(\gamma_2,\delta_2)}(\homatinfty{p_2},\cohomatinfty{q_2})=e^{-2\pi i(\langle\homatinfty{p_2},\cohomatinfty{\delta_1}\rangle+\langle\homatinfty{\gamma_1},\cohomatinfty{q_2}\rangle)}.
\end{equation}
The advantage is that we no longer need to think about the relative spatial orientation with respect to an arbitrarily chosen forbidden cone. It is also possible to restore the symmetry between $1$ and $2$ by introducing a similar pair $(\homatinfty{p_1},\cohomatinfty{q_1})$ of classes at infinity. In this case we need to subtract the terms obtained by exchanging $1$ and $2$ in the exponent. Moreover, eq. \eqref{eq:braidingatinfinity} (as well as its symmetric variant) also makes sense in our general setting without the special form of $E$.

One should note however, that if this quantity is to have any physical meaning, the classes appearing in the subscript and in the argument need to be related to each other. Otherwise we could even get a nontrivial braiding between two vacuum representations, which is, at best, difficult to interpret. In general, we do not know what relations need to be required, but at least on $F\times \mathbb{R}^2$ with $F$ compact the above arguments single out a ``canonical'' choice.

Now we return to the general case, and consider conjugate charges. A conjugate of the endomorphism $\rho$ is a triple $(\bar{\rho},R,\bar{R})$ where $R\in\Hom(\id,\bar{\rho}\otimes\rho)$ and $\bar{R}\in\Hom(\id,\rho\otimes\bar{\rho})$ satisfying $\bar{R}^*\rho(R)=I=R^*\bar{\rho}(\bar{R})$. For the endomorphism $\rho=\rho_{(\gamma,\delta)}$ we can choose $\bar{\rho}=\rho_{(-\gamma,-\delta)}$ and $R=\bar{R}=I$. The twist of an endomorphism $\rho$ is then defined as \cite{Halvorson}
\begin{equation}\label{eq:twist}
\Theta_\rho=(\bar{R}^*\otimes\id_\rho)\circ(\id_{\bar{\rho}}\otimes\varepsilon_{\rho,\rho})\circ(\bar{R}\otimes\id_\rho),
\end{equation}
where the arguments of the braiding are suppressed. When $\rho=\rho_{(\gamma,\delta)}$ the representation $\pi_0\circ\rho$ is irreducible, therefore $\Theta_\rho$ is just a phase factor corresponding to the statistics of the sector. In this case using the conjugates as above one can see that the phase factor is the same as that appearing in $\varepsilon_{\rho,\rho}$. We remark that in the special case when the braiding is computed for an endomorphism with itself, the symmetrized version only depends on the differences of the (co-)homology classes at infinity appearing in the argument. As before, we do not know if there is a canonical choice for these when the space is not of the form $F\times\mathbb{R}^2$.

\section{Acknowledgement}
We would like to thank Pieter Naaijkens for helpful correspondence.


\begin{thebibliography}{10}

\bibitem{Kitaev}
A.~Y. Kitaev, ``Fault-tolerant quantum computation by anyons,'' {\em Annals of
  Physics}, vol.~303, no.~1, pp.~2--30, 2003.

\bibitem{FreedmanMeyerLuo}
M.~H. Freedman, D.~A. Meyer, and F.~Luo, ``Z2-systolic freedom and quantum
  codes,'' {\em Mathematics of quantum computation, Chapman \& Hall/CRC},
  pp.~287--320, 2002.

\bibitem{FreedmanHastings}
M.~H. Freedman and M.~B. Hastings, ``Quantum systems on non-k-hyperfinite
  complexes: A generalization of classical statistical mechanics on expander
  graphs,'' {\em Quantum Information \& Computation}, vol.~14, no.~1-2,
  pp.~144--180, 2014.

\bibitem{BravyiHastings}
S.~Bravyi and M.~B. Hastings, ``Homological product codes,'' in {\em
  Proceedings of the 46th Annual ACM Symposium on Theory of Computing},
  pp.~273--282, ACM, 2014.

\bibitem{topmemory}
E.~Dennis, A.~Kitaev, A.~Landahl, and J.~Preskill, ``Topological quantum
  memory,'' {\em Journal of Mathematical Physics}, vol.~43, no.~9,
  pp.~4452--4505, 2002.

\bibitem{Brell}
C.~G. Brell, ``A proposal for self-correcting stabilizer quantum memories in 3
  dimensions (or slightly less),'' {\em arXiv preprint arXiv:1411.7046}, 2014.

\bibitem{DHR1}
S.~Doplicher, R.~Haag, and J.~E. Roberts, ``Local observables and particle
  statistics {I},'' {\em Communications in Mathematical Physics}, vol.~23,
  no.~3, pp.~199--230, 1971.

\bibitem{DHR2}
S.~Doplicher, R.~Haag, and J.~E. Roberts, ``Local observables and particle
  statistics {II},'' {\em Communications in Mathematical Physics}, vol.~35,
  no.~1, pp.~49--85, 1974.

\bibitem{Halvorson}
H.~Halvorson and M.~M{\"u}ger, ``Algebraic quantum field theory,'' {\em arXiv
  preprint math-ph/0602036}, 2006.

\bibitem{NaaijkensLocalized}
P.~Naaijkens, ``Localized endomorphisms in {K}itaev's toric code on the
  plane,'' {\em Reviews in Mathematical Physics}, vol.~23, no.~04,
  pp.~347--373, 2011.

\bibitem{NaaijkensDuality}
L.~Fiedler and P.~Naaijkens, ``Haag duality for {K}itaev's quantum double model
  for abelian groups,'' {\em arXiv preprint arXiv:1406.1084}, 2014.

\bibitem{NaaijkensIndex}
P.~Naaijkens, ``{K}osaki-{L}ongo index and classification of charges in 2{D}
  quantum spin models,'' {\em Journal of Mathematical Physics}, vol.~54, no.~8,
  p.~081901, 2013.

\bibitem{AFH}
R.~Alicki, M.~Fannes, and M.~Horodecki, ``A statistical mechanics view on
  {K}itaev's proposal for quantum memories,'' {\em Journal of Physics A:
  Mathematical and Theoretical}, vol.~40, no.~24, p.~6451, 2007.

\bibitem{ends}
B.~Hughes and A.~Ranicki, {\em Ends of complexes}.
\newblock No.~123, Cambridge university press, 1996.

\bibitem{Hatcher}
A.~Hatcher, {\em Algebraic Topology}.
\newblock Cambridge University Press, 2002.

\bibitem{HiltonStammbach}
P.~J. Hilton and U.~Stammbach, {\em A Course in Homological Algebra (Graduate
  Texts in Mathematics, 4)}.
\newblock Berlin-Heidelberg-New York: Springer, 1971.

\bibitem{BR1}
O.~Bratteli and D.~W. Robinson, {\em Operator Algebras and Quantum Statistical
  Mechanics 1.}
\newblock Texts and Monographs in Physics, Springer, second~ed., 1987.

\bibitem{BR2}
O.~Bratteli and D.~W. Robinson, {\em Operator Algebras and Quantum Statistical
  Mechanics 2.}
\newblock Texts and Monographs in Physics, Springer, second~ed., 1997.

\bibitem{CNN}
M.~Cha, P.~Naaijkens and B.~Nachtergaele, ``The complete set of infinite volume ground states for {K}itaev's abelian quantum double models,'' {\em arXiv preprint arXiv:1608.04449}, 2016.

\bibitem{KadisonRingrose}
R.~V. Kadison and J.~R. Ringrose, {\em Fundamentals of the theory of operator
  algebras. 2: Advanced theory}.
\newblock Academic press, 1986.

\bibitem{Takesaki}
M.~Takesaki, {\em Theory of operator algebras {I}. Reprint of the first (1979)
  edition. Encyclopaedia of Mathematical Sciences, 124. Operator Algebras and
  Noncommutative Geometry, 5}.
\newblock Springer-Verlag, Berlin, 2002.

\bibitem{EvansKawahigashi}
D.~E. Evans and Y.~Kawahigashi, {\em Quantum symmetries on operator algebras},
  vol.~147.
\newblock Clarendon Press Oxford, 1998.

\end{thebibliography}
\end{document}